\documentclass[runningheads]{llncs}
\usepackage{graphicx}
\usepackage[
backend=biber,
style=numeric,
giveninits=true
]{biblatex}
\DeclareNameAlias{default}{family-given}
\usepackage[utf8]{inputenc}
\usepackage{verbatim}
\usepackage{bbm}
\usepackage[toc,page,title]{appendix}
\usepackage{xcolor}

\graphicspath{ {figures/} }
\usepackage{tikz}
\usetikzlibrary{positioning,chains,fit,shapes,calc}
\usepackage{blindtext}
\usepackage{amsmath}
\usepackage{amssymb}
\usepackage[english]{babel}
\usepackage{fancyvrb}
\usepackage{enumerate}
\usepackage{mathrsfs}
\usepackage{mathtools}

\mathchardef\mhyphen="2D

\usepackage[ruled,noline,linesnumbered, noend]{algorithm2e}
\newtheorem{claim1}[theorem]{Claim}

\newif\ifarxiv
\arxivfalse 

\let\oldnl\nl
\newcommand{\nonl}{\renewcommand{\nl}{\let\nl\oldnl}}

\let\C\relax
\DeclareMathOperator{\C}{\mathcal{C}} 
\let\D\relax
\DeclareMathOperator{\D}{\mathcal{D}} 
\let\P\relax
\DeclareMathOperator{\P}{\mathcal{P}} 
\let\E\relax
\DeclareMathOperator{\E}{\mathcal{E}} 
\let\S\relax
\DeclareMathOperator{\S}{\mathsf{state}}
\let\Pr\relax
\DeclareMathOperator{\Pr}{\mathbb{P}}
\let\negl\relax
\DeclareMathOperator{\negl}{\mathsf{negl}}

\begin{document}

\title{Achieving Almost All Blockchain Functionalities with Polylogarithmic Storage}



\author{Parikshit Hegde\inst{1} \and
Robert Streit\inst{1} \and
Yanni Georghiades\inst{1} \and
Chaya Ganesh\inst{2} \and
Sriram Vishwanath\inst{1}}
\authorrunning{P. Hegde et al.}
\titlerunning{Achieving Almost All Blockchain Functionalities with Polylog. Storage}
%
\institute{The University of Texas at Austin, TX, USA 
\email{\{hegde, rpstreit, yanni.georghiades, sriram\}@utexas.edu} \and
Indian Institute of Science, KA, India\\
\email{chaya@iisc.ac.in}}


\maketitle

\begin{abstract}
 In current blockchain systems, full nodes that perform all of the available functionalities need to store the entire blockchain. In addition to the blockchain, full nodes also store a blockchain-summary, called the \emph{state}, which is used to efficiently verify transactions. With the size of popular blockchains and their states growing rapidly, full nodes require massive storage resources in order to keep up with the scaling. This leads to a tug-of-war between scaling and decentralization since fewer entities can afford expensive resources. We present \emph{hybrid nodes} for proof-of-work (PoW) cryptocurrencies which can validate transactions, validate blocks, validate states, mine, select the main chain, bootstrap new hybrid nodes, and verify payment proofs. With the use of a protocol called \emph{trimming}, hybrid nodes only retain polylogarithmic number of blocks in the chain length in order to represent the proof-of-work of the blockchain. Hybrid nodes are also optimized for the storage of the state with the use of \emph{stateless blockchain} protocols. The lowered storage requirements should enable more entities to join as hybrid nodes and improve the decentralization of the system. We define novel theoretical security models for hybrid nodes and show that they are provably secure. We also show that the storage requirement of hybrid nodes is near-optimal with respect to our security definitions.
 
\end{abstract}


\section{Introduction}
\label{ssec:introduction}

Blockchains enable a group of untrusting parties to securely maintain a distributed ledger without relying on a trusted third party. 
Instead, the power to decide what is recorded in the blockchain is distributed amongst a set of decentralized nodes. 
This property is desirable for applications used by a set of mutually distrustful parties, such as a digital currency.
For this reason, \emph{cryptocurrencies} are a fundamental application of blockchains and are increasingly growing in popularity. In this paper, we focus on cryptocurrencies built on top of a proof-of-work (PoW) blockchain employing the longest chain rule.
A blockchain node for a cryptocurrency typically has the following functionalities:

\begin{enumerate}
    \item \emph{Transaction validation}: When the node receives a new transaction, it checks if the transaction is valid with respect to the transactions already confirmed in the blockchain.
    \item \emph{Block validation}: When the node receives a new block, it verifies that the block hash is valid, all the transactions in the block are valid, and the block otherwise follows all of the conventions imposed by the protocol. 
    \item \emph{State Validation}: Given a summary of currency ownership in the system, called the \emph{state}, the node verifies that the state is consistent with the \\ 
    blockchain.
    \item \emph{Mining}: The node can append a block to the blockchain by verifying its contents and producing a PoW for the block. A node which does not mine blocks is assumed to have mining power 0.
    \item \emph{Chain Selection}: Given a set of conflicting chains, the node can choose the main chain which has the most PoW. 
    Any two honest nodes which receive the same set of conflicting chains in the same order must select the same main chain.
    \item \emph{Bootstrapping new nodes to the blockchain}: A node can provide new nodes entering the system with the blockchain.
    \item \emph{Serving payment proofs}: Given a transaction $tx$, the node can provide a \emph{proof} of $tx$'s inclusion in the blockchain.
    \item \emph{Verification of payment proofs}: Given a proof of transaction $tx$'s inclusion in the blockchain, the node can verify the correctness of the proof.
\end{enumerate}

In order to have all of the functionalities above, a node must verify and store the entire blockchain. 
We will refer to such nodes as \emph{full nodes}. 
Popular systems like Bitcoin and Ethereum also allow for other types of nodes with more limited functionalities \cite{bitcoinfullnode, ethereumnode}. 
For instance, \emph{pruned nodes} initially download the entire blockchain and verify it. 
However, they later \emph{prune} the blockchain, meaning that they discard block data and only retain block headers for blocks older than the most recent $k$ blocks in the blockchain. 
By retaining a summary of the blockchain called the \emph{state}, they can still perform all desired functionalities except for serving proofs of payment and bootstrapping new nodes. \emph{Lightweight nodes} only download the block-headers, and their only functionality is to verify payment proofs provided by full nodes. Importantly, full nodes are necessary to bootstrap both pruned and lightweight nodes into the blockchain.

Since full nodes need to store the entire blockchain, their resource requirements can be high. 
This is an entry barrier that makes fewer nodes participate, which leads to a centralization of trust. 
In this paper, we optimize storage requirements in order to lower this entry barrier. 
In deference to the storage capabilities of modern computational hardware, we divide storage into two categories. 
The first is \emph{cold storage}, which is accessed infrequently and is stored on disk. This includes older blocks that are deep inside the blockchain. The second is \emph{hot storage}, which is accessed frequently and is stored in memory. Naturally, the blockchain state used to validate blocks and transactions is kept in hot storage.

In this paper, we propose a new class called \emph{hybrid nodes}. Hybrid nodes have all of the above functionalities except for the ability to provide payment proofs. 
Importantly, hybrid nodes can bootstrap new hybrid nodes into the system, meaning they do not depend on any other type of nodes, including full nodes. 
Moreover, if $B$ is the length of the blockchain, hybrid nodes only require $polylog(B)$ cold storage to represent the PoW of the chain. 
This is achieved by a process we call \emph{trimming}, an extension of non-interactive-proofs-of-proof-of-work \cite{kiayias2020} (henceforth, NIPoPoW).
NIPoPoW is a protocol that enables a prover (which is most often a full node) to provide payment proofs of $polylog(B)$ size rather than the traditional $B$ size, but NIPoPoW still requires the prover to store the entire blockchain.
We extend these techniques further in our trimming protocol to securely remove blocks and reduce storage. 

We now comment on the \emph{practical implications} of our proposed protocols for hybrid nodes. There are two main components of a blockchain with significant storage requirements for hybrid nodes. First is the storage required to represent the PoW of the chain, which is used by the consensus protocol. In traditional systems, since the entire chain of block-headers must be stored, the storage requirements for this component at the time of writing could be in the order of 100s of megabytes for systems such as Bitcoin and Ethereum. Our trimming protocols for hybrid nodes can decrease this requirement to the order of 100s of kilobytes. While 100MB might not seem large, if one wishes to run a number of blockchains on a single device then the storage requirement can quickly multiply into the gigabytes range if methods such as \emph{trimming} are not employed. Moreover, since hybrid nodes only store $polylog(B)$ number of block headers, their storage requirement grows slower with time too. The second component that requires storage is the blockchain state (UTXO or account-based for cryptocurrencies). For instance, the size of Bitcoin's UTXO set is roughly 4GB \cite{bitcoinutxo}. However, some novel stateless blockchain protocols reduce this storage requirement to the order of kilobytes by requiring clients to provide payment proofs \cite{agrawal2020}. In Section \ref{sec:stateless_blockchains} and Appendix \ref{sec:app-stateless_blockchains}, we show that hybrid nodes can employ stateless blockchain protocols, thus optimizing both their PoW and state storage.

Previous works, specifically CoinPrune and SecurePrune \cite{matzutt2020, reddy2021}, achieve the same functionalities as hybrid nodes with lower storage requirements than full nodes. 
They achieve this by storing a commitment to the blockchain-state in the blocks and pruning blocks that are deep in the blockchain. 
However, their storage requirement still scales linearly with blockchain length.
Moreover, these works provide a largely qualitative analysis of their respective protocols. 
In contrast, we perform a rigorous security analysis and provide proofs that hybrid nodes are secure.

Concurrent with the initial submission of our work, we were made aware of an independent work that uses a modification of NIPoPoWs to obtain polylogarithmic storage \cite{Kiayias2021Mining}. Although both protocols are similarly motivated, we believe that our security definitions and the corresponding analysis are novel and crucial to this area. Of particular note, we believe that security against a trim-attack (see Section \ref{sec:security-statements}), is crucial for the operation of hybrid nodes. Unlike our protocol, \cite{Kiayias2021Mining} claim to not require \emph{optimism} for \emph{succinctness} (see further in Theorem \ref{th:succinctness} and Remark \ref{rm:need_optimism}). However, we note that it doesn't seem economically viable for an adversary to expend resources to simply hurt the succinctness of hybrid nodes.

We defer the discussion of other related works which aim to optimize \newline
blockchain storage to Appendix \ref{app:other_work}.

We now summarize our results and outline the organization of the paper. In section \ref{sec:model}, we introduce the basic model and notation. In section \ref{sec:preliminaries}, we summarize CoinPrune and NIPoPoW, which are building blocks for our protocol. In section \ref{sec:trimming}, we explain the properties of the hybrid node's chain and describe the trimming protocol, though we defer the description of chain selection and state verification protocols to Appendix \ref{app:protocol} due to space constraints. In section \ref{sec:security-statements}, we introduce novel security definitions for hybrid nodes, including $\mathsf{trim\mhyphen attacked}$, $\mathsf{congruence}$, $\mathsf{state\mhyphen attacked}$, and $\mathsf{bootstrap\mhyphen attacked}$. In section \ref{sec:theorems}, we show that hybrid nodes satisfy all the security properties with high probability, and we also discuss the polylogarithmic storage requirement and the lower bound on the storage requirement. For brevity, formal proofs for these results are deferred to Appendices \ref{app:security-proof}-\ref{app:succinctness}. In section \ref{sec:stateless_blockchains}, we illustrate that when combined with stateless blockchain protocols, hybrid nodes are optimized both in terms of cold and hot storage. And, we examine directions for future work in Section \ref{app:future}.

Our primary contributions are the protocols associated with the hybrid nodes and the novel security definitions and their associated theorems. These are in sections \ref{sec:trimming}, \ref{sec:security-statements} and \ref{sec:theorems}.







\section{Model and Notation}
\label{sec:model}
In this work, we consider a set of nodes running a PoW blockchain.  We model the system using continuous time, which accurately models systems with high hash-rates such as Bitcoin and Ethereum\cite{bitcoinhash,ethereumhash}. In this section, we restrict the model description to the essentials required to describe our protocol. Supplemental model details used for security analysis are deferred to Appendix \ref{app:model_supp}. 

Several communication models are considered in the literature. The simplest is the \emph{synchronous} model where a block broadcast by a node at a certain time is received by all other nodes immediately \cite{garay2015, eyal2014}. Since time is continuous, no more than one block is mined at any given time, implying only one block could be in communication at any given time. 
More complicated communication models with communication delays are also considered in the literature \cite{pass2017, sankagiri2021}. For the sake of simplicity, we consider the synchronous model in this paper and leave it to future work to transfer our results to more complicated communication models. Note that because of synchronous communication, all honest nodes have knowledge of the same set of blocks at any given time.

\subsubsection*{Basic Blockchain Notation}
The honest (longest-)chain at time $t$ is represented by $\C_t$. 
The number of blocks in $\C_t$ is called the chain-length and is denoted as $B_t$. 
When the time $t$ is clear from context, we may drop the subscript and refer to it as just $\C$. 
Blocks in $\C$ are indexed as an array in a similar convention to the Python programming language, meaning that $\C[i]$ is denoted as block $i$. Since it is convenient, we refer to a block by its index $i$ and not its contents. $\C[0]$ is called the  \emph{genesis block}. $\C[i_1:i_2]$ represents the segment of the chain from block $i_1$ to block $(i_2-1)$. If at any time an honest node hears of another chain $\D$ which is longer than $\C_t$, then it adopts $\D$ as the honest chain (i.e., $\C_{t^+} = \D$, where $t^+$ indicates the time incrementally after $t$). The last common block between two chains $\C$ and $\D$ is called the \emph{latest common ancestor} (LCA), and is denoted as $b=\mathsf{LCA}(\C, \D)$.
Specifically, $\C_t[:b+1] = \D[:b+1]$, and $\C_t[b+1:] \cap \D[b+1:] = \emptyset$. When a new block $\Bar{b}$ is appended to $\C$, we denote the extended chain as $\C\Bar{b}$.

In our model, hybrid nodes do not store the entire chain $\C_t$, but instead store a trimmed version which contains fewer blocks than $\C_t$. The trimming protocol and its associated notation is described in Section \ref{sec:trimming}.

\subsubsection*{Blockchain State}
When a new transaction is submitted, a node must check if it is ``valid" with respect to the chain $\C_t$. This could be accomplished by parsing through the complete log of transactions in $\C_t$. However, due to the rapidly increasing size of $\C_t$, it is far more efficient for a node to validate transactions against a summary of the chain called the state, and denoted as $\S(\C_t)$. Equivalently, we may refer to the state as $\S(B_t)$, where $B_t$ is the length of the chain $\C_t$. 
Validating a transaction against $\C_t$ is equivalent to validating it with respect to $\S(\C_t)$, so using $\S(\C_t)$ is preferred due to its smaller hot storage requirement. After a new block $\Bar{b}$ is added to the chain, the new state is computed as $\S(\C_t \Bar{b}) = F(\S(\C_t),\Bar{b})$, where $F()$ is a function that applies the transactions in $\Bar{b}$ to $\S(\C_t)$. When the new block is clear from context, we denote the function simply as $F(\S(\C_t))$, and when the function is applied on $n$ sequential new blocks, we denote the operation as $\S(\C_tb_1b_2\dots b_n) = F^{n}(\C_t)$. Two types of states are popular: 1) \emph{UTXO-based State}: this stands for unspent transaction output, and is used by Bitcoin. The UTXO state consists of a list of unspent coins. A new transaction is valid with respect to the state, if it consumes one or more of these coins, and creates new coins whose total value is no larger than the consumed coins; 2) \emph{Account Based State: } This is used in Ethereum. An account-based state consists of a vector of key-value mappings, with one mapping corresponding to each user. The key establishes the user's identity, and the value establishes the balance in the user's account. A user can issue a transaction that transfers a part of his account's balance to another user.

\subsubsection*{Interlinks}

\begin{figure*}[t]
    \centering
    \includegraphics[width=\linewidth]{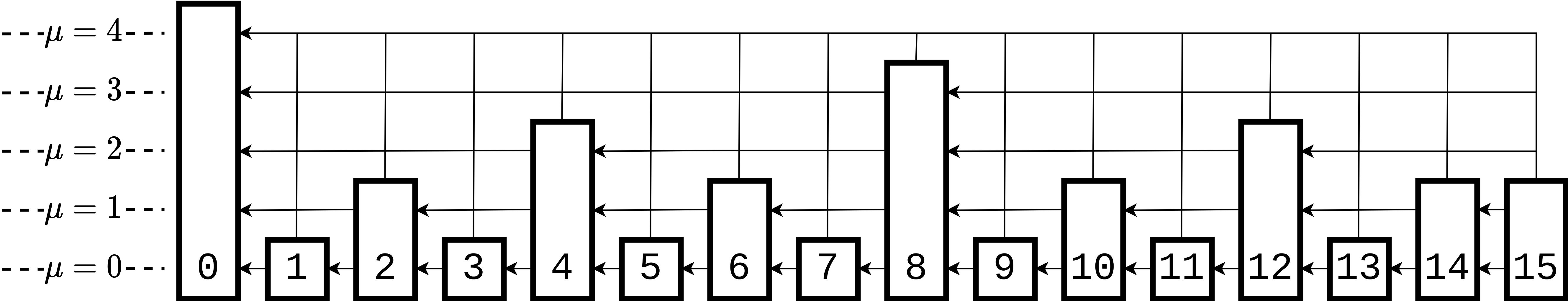}
    \centering
    \caption{An example of the interlink structure (inspired by Figure 1 of \cite{Daveas2020}). At the bottom of each block is its index, and the block's height signifies its superblock level. Each block has a link to the closest ancestor at every level, which is shown by the arrows linking blocks to ancestors. Notice all the blocks are contained in level 0, and only the genesis block is in level 4.}
    \label{fig:interlinks}
\end{figure*}

In traditional blockchains like Bitcoin, each block contains a link (using a hash) to the previous block in the chain. To enable hybrid nodes to store the blockchain in a succinct way, we employ a clever link structure called the \emph{interlink}. Interlinks were introduced in \cite{kiayias2016} and further developed in \cite{kiayias2020}.

In the interlink model, a block contains the following information: 1) transactions in the block; 2) the Merkle root $x$ of all the transactions in the block; 3) the Merkle root $y\left(\S(B_t)\right)$ of the corresponding blockchain state; 4) the block index $i$; 5) the \emph{interlink}, which contains hash links to several previous blocks and is described in detail in the following paragraphs; 6) the random nonce $\eta$; and 7) the block hash $\mathtt{id} = H(\eta, x, y, i, \mathtt{interlink} )$, where $H()$ is a hash function. For a block to be valid, $\mathtt{id}$ must contain at least $T$ leading $0$'s. Equivalently, we say that $\mathtt{id} \leq 2^{-T}$. All the information in the block \emph{except} the list of transactions is referred to as the \emph{block-header}. Observe that the $\mathtt{id}$ of the block can be verified given just the block-header.

To describe the interlink, we first need to define \emph{superblocks}. A level-$\mu$ superblock is a block with $\mathtt{id} \leq 2^{-(T+\mu)}$. Since a valid block satisfies $\mathtt{id} \leq 2^{-T}$, all valid blocks are level-0 superblocks. The genesis block is defined to be a superblock of every level from $0$ to $\infty$. And, a level-$\mu$ superblock is also a level-$\mu'$ superblock for all $0\leq \mu' \leq \mu$, since $2^{-(T+\mu)} \leq 2^{-(T+\mu')}$.

The $\mathtt{interlink}$ data-structure in a block contains a link to the previous superblock of level $\mu$ for every level $\mu$ that is in the chain $\C_t$ up to that block. Since the previous block will always be a superblock of level at least $0$, the interlink always contains a link to the previous block (thus, without any further modification, the security properties of the blockchain are unaffected). Also, since the genesis block is of all possible levels, a link to the genesis block is always included. A pictorial example of this is shown in Figure \ref{fig:interlinks}.

Using interlinks, it is possible to ``skip" over blocks when traversing the blockchain. To be more specific, it is useful to define notation for ``traversing the blockchain at level-$\mu$''. For any given chain $\C$, the level-$\mu$ \emph{upchain}, denoted $\C\uparrow^\mu$, is the sequence of all level-$\mu$ superblocks in $\C$. That is,
\begin{equation}\label{eq:upchain}
    \C\uparrow^\mu \triangleq \left\{ b: b \in \C, \text{ and } \mathtt{id}(b) \leq 2^{-(T+\mu)} \right\}.
\end{equation}

Note that although it is convenient to use set-notation to define it, $\C\uparrow^\mu$ is a sequence with the order of its blocks being the same as they are in $\C$. From the definition of the interlink, each block in the upchain $\C\uparrow^\mu$ contains a reference to the previous block in the upchain. Therefore, it is possible to traverse through $\C\uparrow^\mu$. Additionally, a chain $\C'$ is called a level-$\mu$ \emph{superchain} if all its blocks are level-$\mu$ superblocks. That is, if the underlying chain of $\C'$ is $\C$, then $\C' \subseteq \C\uparrow^{\mu}$.



We use square-brackets to index $\C\uparrow^{\mu}$, similar to a python array. However, at times it is useful to refer to blocks in $\C\uparrow^{\mu}$ according to the block's index in $\C$. In this case, we use curly-braces to index $\C\uparrow^{\mu}$. This is best illustrated using an example. Consider, $\C = \{ 0, 1, 2, 3, 4, 5, 6, 7, 8, 9, 10 \}$, and let $\C\uparrow^{\mu} = \{0, 3, 4, 7, 10\}$. In this case, $\C\uparrow^{\mu}[3:] = \{7, 10\}$, but $\C\uparrow^{\mu}\{3:\} = \{3,4,7,10\}$. As a further illustration of the upchain notation, in the blockchain in Figure \ref{fig:interlinks}, $\C\uparrow^{2} =\{ 0, 4, 8,12\}$.

\section{Preliminaries}
\label{sec:preliminaries}

Before describing our protocol, we summarize CoinPrune \cite{matzutt2020} and NIPoPoW \cite{kiayias2020}. Our protocol builds upon both of these protocols.

\subsubsection*{CoinPrune}
In CoinPrune, similar to our model, blocks contain commitments to the blockchain-state. The protocol selects a \emph{pruning point} $k$ blocks from the tip of the chain.
All blocks after the pruning point are retained completely, while only block headers are retained prior to the pruning point. The state of the blockchain at the pruning point is also stored. Then, the PoW in the chain can be established since the block headers for the entire chain are preserved. Furthermore, since the state at the pruning point is preserved and a commitment to it is stored in the blocks, the state corresponding to each block after the pruning point can be recovered. Refer to Figure \ref{fig:coinprune} for a visual description.

\begin{figure}[h]
  \centering
  \includegraphics[width=0.5\linewidth]{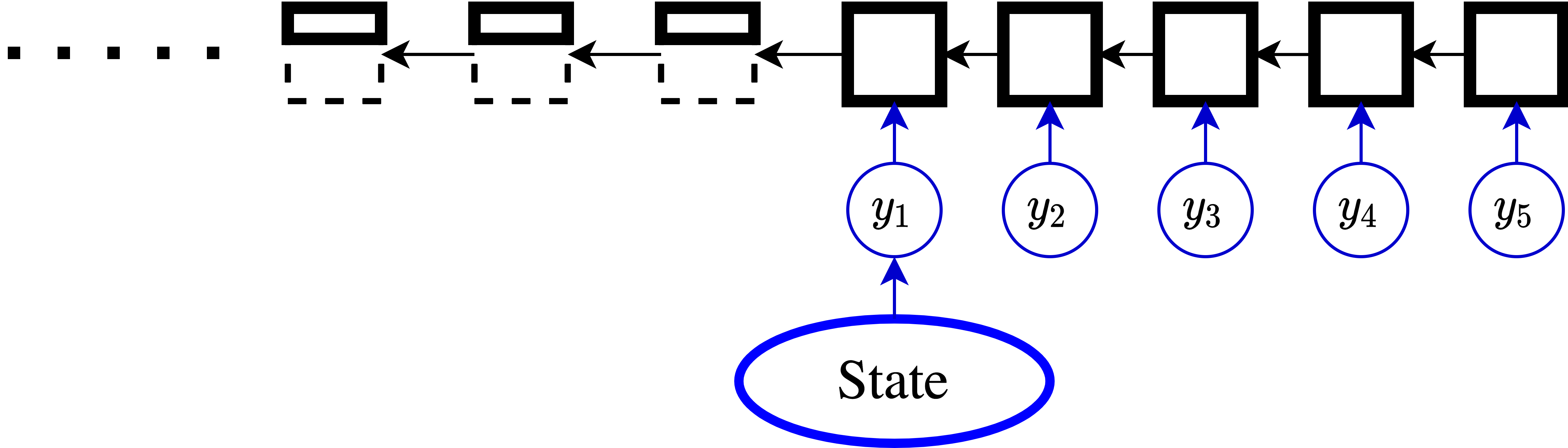}
  \caption{An instance of CoinPrune.  Complete square-boxes indicate the complete blocks after the pruning point, and incomplete boxes indicate block-headers of blocks before the pruning point. The variables $y_i$ are the state-commitments of the respective states stored in the block headers. They are shown separately from the block only for emphasis. The state at the pruning point is stored, and its validity is confirmed by the state-commitments $y_i$'s.}
  \label{fig:coinprune}
\end{figure}

While our method for establishing the state of the blockchain is similar to CoinPrune, we deviate in the way we establish the PoW of the chain. We note that although CoinPrune improves storage compared to full nodes, they still store the entire block-header chain prior to the pruning point. Thus, their storage requirement still scales linearly in the length of the chain, albeit the multiplicative constant may be very small. In contrast, we retain only a subset of the block headers, leading to sublinear storage requirement in the number of blocks in the blockchain. 
To accomplish this, we take inspiration from NIPoPoW.

\subsubsection*{NIPoPoW}
NIPoPoW is a protocol that is used to provide succinct proofs of payments to lightweight clients. Payment proofs have two components. First, the proof needs to establish the amount of PoW in the blockchain, and second, it needs to contain a proof of inclusion of the payment in the chain. Traditionally, the PoW of the chain is established by sending the entire chain of block headers to the lightweight client. NIPoPoWs optimize this step by making the following observation. Informally, by the property of concentration around the means, the $\mu$-upchain $\C\uparrow^{\mu}$ of an underlying chain $\C$ is such that $2^{\mu}\lvert \C\uparrow^{\mu} \rvert \approx  \lvert \C \rvert$ (as an illustration, in Figure \ref{fig:interlinks}, 
      superblocks of level-2 appear roughly every $2^2 = 4$ blocks). Recall that a level-$\mu$ superblock is $2^\mu$ times harder to find than a regular block (i.e., a level-0 superblock). Therefore, if $2^{\mu}\lvert \C\uparrow^{\mu} \rvert \approx  \lvert \C \rvert$, then it is as hard for an adversary to create a fork around $\C\uparrow^{\mu}$ with level-$\mu$ superblocks as it is to create a fork around $\C$ with level-$0$ superblocks. Therefore, it is sufficient to just provide $\C\uparrow^{\mu}$ as a proof of the PoW of the chain. For a large enough level $\mu$, $\C\uparrow^\mu$ is much smaller in size than the underlying chain $\C$, thus making NIPoPoWs much faster than traditional protocols.

Our protocol differs from NIPoPoWs in several ways. First, the goal of our protocol is to optimize a hybrid node's storage while retaining almost all of a full node's functionalities, while the goal of NIPoPoWs is to provide succinct payment proofs. Second, NIPoPoWs do not optimize the prover's storage since the prover must still store the entire blockchain. Third, NIPoPoWs are one-time proofs of payment, meaning they need to be generated afresh for every new proof request, whereas our protocol proceeds in an iterative manner throughout the blockchain's execution. In particular, we employ different \emph{level ranges} (elaborated in the next section) in order to optimize storage throughout time, whereas NIPoPoW only uses a single level range. Fourth, since our end goal is different, our security requirements are different from NIPoPoWs. Lastly, we note that since our security models are different, we use vastly different parameters in our protocol compared to NIPoPoW, and also do novel analysis.


\section{Trimming Protocol}
\label{sec:trimming}

In this section we describe our protocol to trim the blockchain. Other associated protocols that compare trimmed-chains to select the main chain, and that verify the blockchain state are described in Appendix \ref{app:protocol} due to lack of space. 

 
First, we start with an intuitive description of the trimmed chain which is best understood by referring to the example in Figure \ref{fig:trimmingexample}. Similar to NIPoPoWs, the high-level idea in our approach is to retain only a subset of (super)blocks in order to represent the proof of work of the chain. Let the trimmed chain be denoted by $\P_t$ (subscript $t$ may be omitted when time is clear from context). $\P$ is a subset of the complete blockchain $\C$, i.e., $\P \subseteq \C$. And, it has an associated  number $B'_t$ called the \emph{trimming point}. All blocks to the right of $B'_t$ are retained, including their data and block headers. That is, $\P\{B':\} = \C[B':]$. We refer to $\P\{B':\}$ as the \emph{untrimmed tail}. Blocks to the left of $B'$ may be trimmed, meaning $\P\{:B'\} \subseteq \C[:B']$. Only the block-headers of the blocks in $\P\{:B'\}$ are retained. The blocks that are not in $\P\{:B'\}$ are permanently deleted by the hybrid node. Here is where we differ from pruning. In pruning, all the block headers are retained. In trimming, blocks are completely deleted, including their headers.
    
\begin{figure*}[t]
  \centering
  \includegraphics[width=\linewidth]{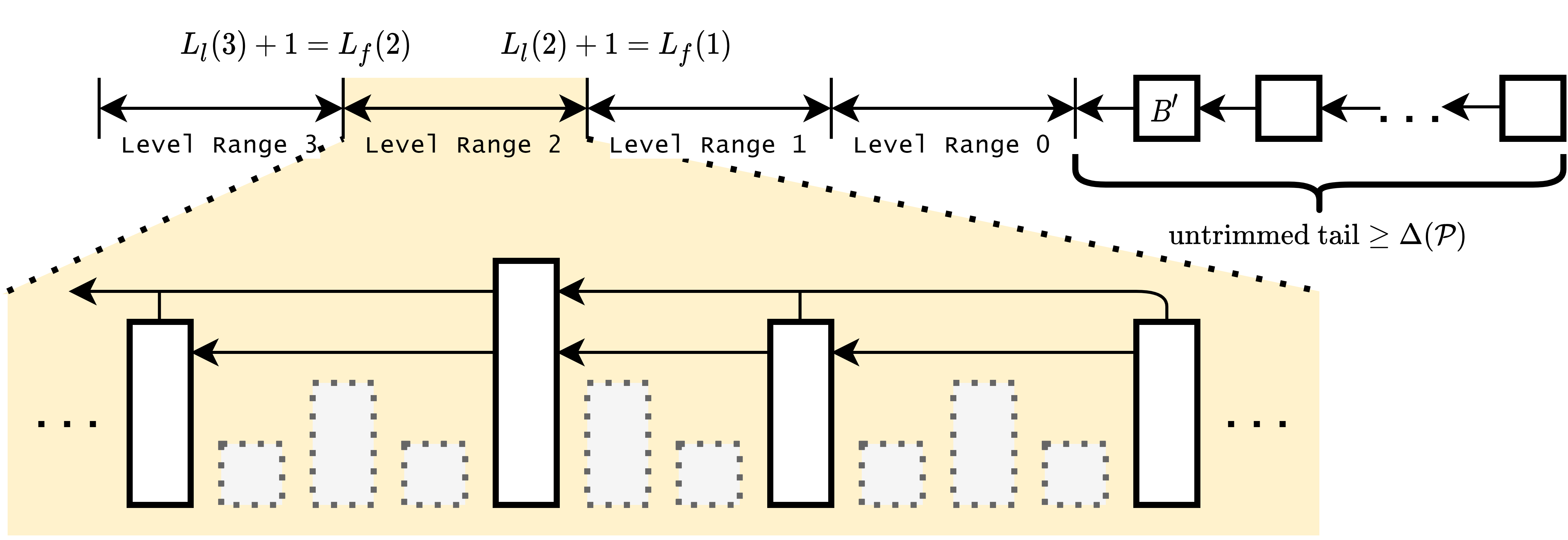}
  \caption{An example of a trimmed chain at an honest node. Notice how it is partitioned into distinct level ranges, each beginning once its predecessor ends. We show a portion of the second level range, $\P\{L_f(2) : L_l(2)\}$. The greyed blocks with dashed outlines represents blocks that have been trimmed. Like Figure \ref{fig:interlinks}, the height of the (super)block is its highest level. Furthermore, after the trimming point, $B'$, is the untrimmed tail where all the blocks are retained.}
  \label{fig:trimmingexample}
\end{figure*}

The trimmed section of $\P$ is further partitioned into \emph{level-ranges}, each level-range corresponding to a unique level $\mu$. A level-range is spread contiguously over a region of the blockchain, and each level range begins at the point its predecessor ends. We denote the starts and ends of level ranges by \emph{level-range functions} $L_f: \mathbb{Z}^+ \to \mathbb{Z}^+$ and $L_l: \mathbb{Z}^+ \to \mathbb{Z}^+$. We define $L_f(\mu)$ as the index of the first block in level-range $\mu$. Similarly, $L_l(\mu)$ is the last block in the level-range $\mu$. Beyond a certain level $\mu_h$, called the highest level, the level-range functions are 0: $L_f(\mu) = L_l(\mu) = 0$ for $\mu > \mu_h$. Also, below a certain level $\mu < \mu_l$, called the lowest level, we have $L_f(\mu) = L_l(\mu) = B'-1$. Notice then that for $\mu_l \leq \mu < \mu_h$, we have $L_l(\mu + 1) + 1 = L_f(\mu)$. The level ranges are also pictorially shown in the example in Figure \ref{fig:trimmingexample}.

     
At level-range $\mu$, we are primarily interested in level-$\mu$ superblocks. As explained with the intuition of NIPoPoW, we need to weigh level-$\mu$ superblocks by $2^\mu$. In order to avoid confusion with PoW, which traditionally does not look at super-levels, we call this notion of weighted PoW simply as the \emph{weight} at level-$\mu$. For the level-range $\mu$, we denote the weight-function $W(\P,\mu)$ as,
    
    \begin{equation*}
        W(\P, \mu) = 2^{\mu} \lvert \P\{L_f(\mu):L_l(\mu)+1\} \uparrow^{\mu} \rvert.
        \label{eq:W}
    \end{equation*}
    
Our trimming protocol in Algorithm  \ref{algo:trimming} ensures that a higher level-range precedes a lower-level range (as illustrated in Figure \ref{fig:trimmingexample}). Therefore, we  can compute the sum of work from the genesis block up to and including level range $\mu$ by, $
    S(\P, \mu) = \sum_{\mu \leq \mu' \leq \mu_h} W(\P, \mu')$.
    
 Since hybrid nodes do not have access to the underlying chain, $S(\P, \mu)$ can be interpreted as their estimate of the PoW up to block $L_l(\mu)$. Note that $W(\P,\mu)$ and $S(\P, \mu)$ are functions of the level-range functions as well. However, we assume that the level-range functions are implicitly defined by $\P$ in order to keep the notation minimal.
 
 \begin{algorithm}[t!]
\DontPrintSemicolon
\SetKwInOut{Input}{input~}
\SetKwInOut{Local}{Local variables~}
\SetKwInOut{Output}{output~}
\DontPrintSemicolon

\Input{
    $\P$ :: My (trimmed) chain\\
    $Q$ :: The trimming interval\\
    $\mu_h$ :: Highest level-range in $\P$ \\
    $L_f(), L_l()$ :: Level-Range functions \\
    $g(), f(), \Delta()$ :: Protocol parameter-functions \\
    $\mathsf{good}_{\delta,g}()$ :: Good-Superchain function (see Appendix \ref{app:protocol})
}

\vspace{\baselineskip}
  
\SetKwFunction{Ftrim}{trim}
\SetKwProg{Fn}{func}{:}{}

\SetKwProg{On}{on event}{:}{}

\On{$\P$ has grown by $Q$ blocks since the last trimming attempt}
{
    $B' = \P[-1] - \Delta(\P)$ \;
    \For{$\mu$ from $\mu_h + 1$ down to 1}
    {   $g \leftarrow g(\P, \mu)$ \;
        $f \leftarrow f(\P, \mu)$ \;
        \If{$|\P\{L_f(\mu):B'\}\uparrow^\mu| \geq f$ and $\mathtt{good}_{\delta, g}(\P\{L_f(\mu):B'\}\uparrow^\mu, \mu)$}
        {
            $\P, L_f, L_l, \text{success} \leftarrow $\Ftrim{$\P$, $L_f$, $L_l$, $B'$, $\mu$, $g$, $f$} \;
            \If{\textnormal{success}=1}
            {
            \textbf{break} \;
            }
        }
    }
    \KwRet \;
}
\;

\Fn{\Ftrim{$\D$, $L_f$, $L_l$, $B'$, $\mu$, g, f}}
{
    $\E \leftarrow \D\{L_f(\mu) : B'\}\uparrow^\mu$ \;
    $A \leftarrow \E[-f]$ \;
    \;
    \For{$\mu'$ from $\mu - 1$ to $0$}
    {
        $\alpha \leftarrow \D\{A : B'\}\uparrow^{\mu'}$ \;
        $\E \leftarrow \E \cup \alpha$ \;
        \;
        \If{$\lvert \alpha \rvert \geq f$ and $\mathsf{good}_{\delta,g}(\alpha, \mu')$}
        {
            $A \leftarrow \alpha[-f]$ \;
        }
    }
    \;
    \eIf{$\lvert \alpha \rvert \geq f$ and $\mathsf{good}_{\delta,g}(\alpha, \mu')$}{
    $L_f(\mu'), L_l(\mu') \leftarrow B'-1$, for all $\mu < \mu'$ \;
    $L_l(\mu) \leftarrow B'-1$ \;
    $\D \leftarrow \D\{:L_f(\mu)\} \cup \E \cup \D\{B':\}$ \;
    \KwRet $\D, L_f, L_l, 1$ \;
    }
    { \KwRet $\D, L_f, L_l, 0$ \;
    }
}
\caption{Trimming Protocol}
\label{algo:trimming}
\end{algorithm} 
 
 The trimming procedure is detailed in Algorithm \ref{algo:trimming}. We briefly describe it here. Trimming is attempted every time the chain grows by $Q$ blocks (line 1). We call $Q$ the trimming interval. The trimming point is set by the required chain-tail length (line 2). Given that $\mu_h$ is the highest level-range in the current trimmed chain, we attempt to trim it further to level $\mu_h+1$. If trimming to that level is not possible, we try to do it to one level lower and so on (line 3). Trimming to a level $\mu$ can be attempted if the specified range of blocks contains enough level-$\mu$ superblocks, and if the corresponding $\mu$-upchain is good (line 6). Roughly, a $\mu$-upchain $\C\uparrow^\mu$ is good if its weight represents the weight of the other level upchains ( $|\C\uparrow^\mu| \approx 2^{(\mu-\mu')} |\C\uparrow^{\mu'}|$, $\mu'<\mu$). For the full definition of \emph{good-superchain}, see Appendix \ref{app:protocol}. Given condition on line 6 is satisfied, the $\mathsf{trim}$ function is called for level $\mu$ (line 7).

The $\mathsf{trim}$ function is very similar to the goodness-aware $\mathsf{Prove}$ algorithm of the NIPoPoW protocol \cite[Algorithm 8]{kiayias2020}. In the $\mathsf{trim}$ function, the $\mu$-upchain is obtained first (line 13). Next, the $(\mu-1)$-level upchain under the last $f$ blocks of the $\mu$-level upchain is also added (lines 17 and 18). If the $(\mu-1)$-level upchain is good, then the $(\mu-2)$-level upchain under its last $f$ blocks is added (lines 20 and 21). Otherwise, the $(\mu-2)$-level upchain under the last $f$ blocks of the $\mu$-upchain are added. This procedure continues until level $0$.

At the end, the $\mathsf{trim}$ function checks if the trimming was a ``success'' by checking if level 0 of the trimmed chain is good. The trim being a success means that it is at least as hard for an adversary to create a longer fork around the trimmed chain, as it would be to do so around the complete chain. Intuitively, this is because the necessary levels of the upchains in this range are good, meaning that they represent the PoW of their corresponding downchains.

In case the trim is a success, the trim function along with the new level range functions are returned (lines 23 to 27) indicating that the trim can be used. Otherwise, the old trimmed chain and level range functions are returned (lines 28 and 29).

 \subsubsection*{Chain Selection} Hybrid nodes need to have a protocol, $\mathsf{Compare}(\C^{(1)}, \C^{(2)})$, to chose the main chain given two conflicting chains, $\C^{(1)}$ and $\C^{(2)}$. Full nodes (that store the entire blockchain) simply choose the longer of the two chains as the main chain. The chain selection protocol for hybrid nodes is a little more complicated since they do not store the entire chain. At a high level, they use the sum of the cumulative weight, $S(\P,0)$, of the trimmed portion of a chain and the length of untrimmed section of the chain as a proxy for the chain length. The complete algorithm involves some more details which are described in full in Appendix \ref{ssec:compare}.
 
 \subsubsection*{State Verification}
Similar to CoinPrune \cite{matzutt2020}, a short commitment to the state at the block is stored in every block. Therefore, a hybrid node can verify the correctness of the blockchain's state by comparing it to the corresponding state commitment. The protocol is called $\mathsf{state-verify}$ and is presented with more details in Appendix \ref{ssec:state_verify}.

\subsubsection*{Hybrid Node's Functionalities}
Here, we describe how a hybrid-node employing the trimming algorithm has the functionalities claimed in Section \ref{ssec:introduction}. Since a hybrid node stores the state of the blockchain at its tip, it can perform transaction validation, block validation and state validation. Using the $\mathsf{Compare}(\cdot,\cdot)$ protocol from Appendix \ref{ssec:compare}, a hybrid node can select the main chain given competing chains. As a consequence, the node can perform mining as well.

A new node joining the system can choose the main chain using the \newline
$\mathsf{Compare}(\cdot,\cdot)$ protocol, and verify the state using the $\mathsf{state\mhyphen verify}$ protocol from Appendix \ref{ssec:state_verify}. Thus, hybrid nodes can bootstrap other hybrid nodes into the system.

The hybrid node verifies payment proofs as follows. In traditional systems like Bitcoin, the prover provides the chain of block-headers in order to establish the PoW, and then provides a short proof for the transaction's inclusion (inclusion-proof) in the chain. In NIPoPoWs the prover provides a superchain (which is logarithmic in the size of the underlying chain) in order to establish the PoW. In either case, the hybrid node can use the $\mathsf{Compare}(\cdot,\cdot)$ protocol to compare the given chain (superchain) to its own trimmed chain. If the two chains only differ near the tail and the inclusion-proof is consistent with the provided chain, then the hybrid node approves the prover's payment proof. Otherwise, it returns false.

\section{Security Definitions }
\label{sec:security-statements}


\subsubsection*{Security from Trim Attack}
Consider the dangerous attack scenario depicted in Figure \ref{fig:trim-attack}, where the adversary provides a trimmed-chain $\P^{(2)}$ which is ``longer'' than the honest chain $\P^{(1)}$ (i.e. $\mathsf{Compare}(\P^{(1)}, \P^{(2)}) = \P^{(2)}$) and the LCA between the two chains precedes the honest chain's trimming point: $b = \mathsf{LCA}(\P^{(1)}, \P^{(2)}) < {B'}^{(1)}$. Denote $b_1$ to be the (super)block after $b$ in $\P^{(2)}$. The honest node cannot verify the state transition from $b$ to $b_1$. This is because the honest node only has access to block headers and not the full state, meaning they can only verify the validity of the block headers, but not of the state transition between the blocks. Moreover, even if the honest node had access to the state at those blocks, there may be a number of blocks between $b$ and $b_1$ that were skipped during the trimming.

\begin{figure}[h]
  \centering
  \includegraphics[width=0.75\linewidth]{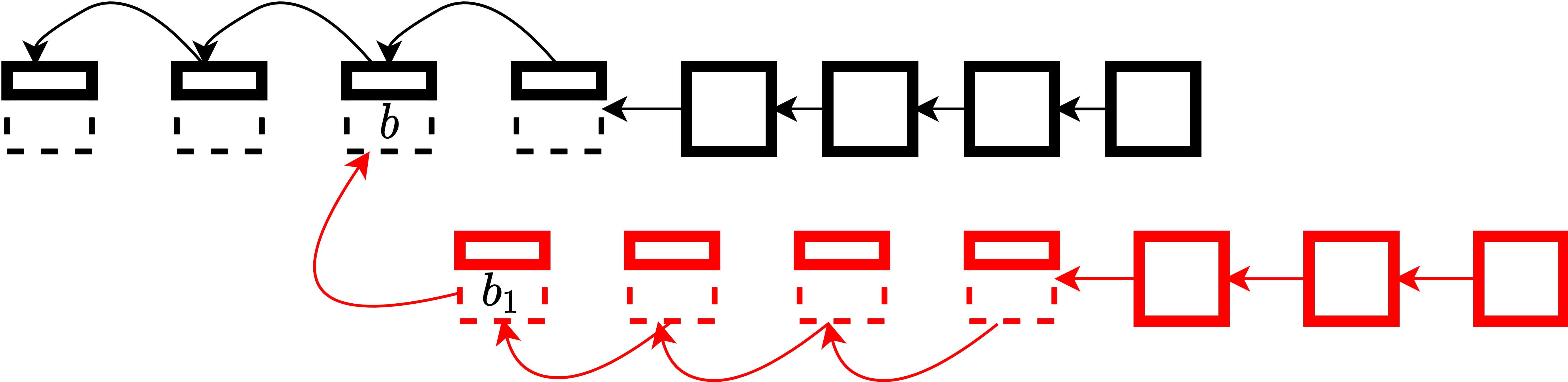}
  \caption{A trim-attack. The honest chain is shown in black, and adversary's chain in red. Complete square-boxes indicate the complete blocks beyond the trimming point, and incomplete boxes indicate block-headers of blocks before the trimming point. Curved arrows indicate that the corresponding blocks may not be subsequent blocks.}
  \label{fig:trim-attack}
\end{figure}

Therefore, if at any time the adversary is able to create a fork as in Figure \ref{fig:trim-attack}, they could arbitrarily alter the state of the chain to their advantage. For instance, the adversary could transfer all of the chain's cryptocurrency into their own accounts. One approach to circumventing this problem is to have the hybrid node re-download the blockchain from a full node in case it encounters a fork preceding its trimming point. However, in that case the security of the system would again rely on the small number of full nodes. We require that hybrid nodes can work independently from full nodes to keep the system as decentralized as possible. To accomplish this, we ensure that there exists no time when an adversary is able to create a fork from beyond a hybrid node's trimming point. 
    \begin{definition}[Attack on the trimmed Chain]
    Let $\omega$ be the fixed randomness\footnote{randomness is w.r.t., the stochastic model described in Appendix \ref{app:model_supp}}.
    Let $\P^{(1)}_{t,\omega}$ be the honest trimmed chain at some honest node at time $t$. Let $\P^{(2)}_{t, \omega}$ be the adversary's trimmed chain. Let $B'(\P^{(1)})$ and $B'(\P^{(2)})$ be their respective trimming points. Let $b_{t,\omega}$ be the LCA block between them: $b_{t, \omega} = \mathsf{LCA}(\P^{(1)}_{t,\omega}, \P^{(2)}_{t, \omega})$. Then, we say that the system is $\mathsf{trim\mhyphen attacked}$ if there exists a time $t$ such that the adversarial trimmed chain is declared to be longer than the honest trimmed chain and the LCA block is before the honest node's trimming point. That is,
    \begin{align*}
        \mathsf{trim\mhyphen attacked} = \{ \omega:& \exists t \text{ s.t. } b_{t, \omega} < B'(\P^{(1)}_{t,\omega}), \text{ and } \mathsf{Compare}(\P^{(1)}_{t, \omega} , \P^{(2)}_{t,\omega}) = \P^{(2)}_{t, \omega} \}.
    \end{align*}
    \label{def:trim_attacked}
    \end{definition}
     \vspace{-.425in}
    \subsubsection*{Congruence}
    Since hybrid nodes only have access to trimmed chains $\P$, we need that the selection of the main trimmed-chain made according to  $\mathsf{Compare}(\cdot,\cdot)$ is in agreement with the underlying complete chains. We formalize this by a property called \emph{congruence}.

    \begin{definition}[Congruence]
    Given any two (complete) chains $\C^{(1)}$ and $\C^{(2)}$ with corresponding trimmed chains $\P^{(1)}$ and $\P^{(2)}$, they are said to be congruent with each other if, $\lvert \C^{(1)} \rvert > \lvert \C^{(2)} \rvert \implies \mathsf{Compare}\left( \P^{(1)}, \P^{(2)} \right) = \P^{(1)}.$
    \label{def:congruence}
    \end{definition}
    
\subsubsection*{State Security}
   Hybrid nodes do not have access to transaction history preceding their trimming point. Instead,
   they rely on the state at the trimming point, $\mathsf{state}(B'_t)$, to compute state at the tip
   of the chain. State is verified using $\mathsf{state\mhyphen verify}$ (Algorithm
   \ref{algo:state-verify}). Firstly, if the adversary launches a trimming attack (Definition
   \ref{def:trim_attacked}), then they could
   change the state of the blockchain arbitrarily. Additionally, the state of the hybrid node is
   also attacked if at any point the adversary can create a different state, $\mathsf{state}'$, that also passes through $\mathsf{state\mhyphen verify}$.

    \begin{definition}[Attack on the State]
    Let $\omega$ be the fixed randomness. At some time $t$, let the trimming point of an honest
    chain $\P^{(1)}_{t,\omega}$ be ${B'}^{(1)}_{t,\omega}$, and let its associated state be
    $\S^{(1)}({B'}^{(1)}_{t,\omega})$. Let the trimming point of the adversary's chain $\P^{(2)}_{t,
    \omega}$ be ${B'}^{(2)}_{t, \omega}$, and let its claimed state be $\S^{(2)}_{t,\omega}$
     The state of the honest node is said to be attacked, denoted by $\mathsf{state\mhyphen
     attacked}$ , if there is either a trim-attack or there exists a time $t$ such that,
     $\S^{(2)}_{t,\omega} \neq \S^{(1)}(B'_{t,\omega})$, such that it verifies against the chain.
     Denoting, $\P_{t,\omega}=\mathsf{Compare}(\P^{(1)}_{t,\omega}, \P^{(2)}_{t,\omega})$, we
     define,
     \begin{align*}
        \mathsf{state\mhyphen attacked} = \{\omega: \omega \in &\mathsf{trim\mhyphen attacked}\} \\
        &\bigcup\{
            \omega: \exists t \text{ such that } \mathsf{state\mhyphen verify}(\S^{(2)}_{t,\omega},
            \P_{t,\omega}) = 1 \}.
    \end{align*}
    \end{definition}

    \subsubsection*{Bootstrapping Security}
    
    When a new (honest) node joins the system, it downloads (possibly trimmed) copies of the blockchain from a number of other hybrid or full nodes. It then chooses the main chain using the $\mathsf{Compare}(\cdot,\cdot)$ protocol and then downloads and verifies the state using the $\mathsf{state\mhyphen verify}$ protocol. In order to show that it adopts the honest trimmed chain and the honest state, we define bootstrapping security below.
    
    \begin{definition}[Bootstrapping Security]
     A joining node is said to be securely bootstrapped into the system if, at the point of it joining, it adopts a trimmed chain that is not $\mathsf{trim\mhyphen attacked}$ with respect to the system's honest chain, and it adopts a state which is not $\mathsf{state\mhyphen attacked}$ with respect to the system's honest state. Otherwise, the node is said to be $\mathsf{bootstrap\mhyphen attacked}$.
     \label{def:bootstrap_attacked}
    \end{definition}
    
    In our system model, we assume that the mining rate of the honest parties and the
    adversary remains constant. This might seem to be counter-factual to Definition
    \ref{def:bootstrap_attacked} because it assumes that nodes can join the system. We remark that
    we make the constant mining rate assumption to make the rigorous security analysis tractable. Practically, we conjecture that if the rate of nodes joining and leaving the system is nearly equal and small enough, then the constant mining rate assumption is a good model for the system.

    \section{Security Results for Hybrid Nodes}
    \label{sec:theorems}
    The protocol has security parameters $k, k' \in \mathbb{N}$, $a, c \in \mathbb{R}^+$, and $\delta \in (0,1)$. Referring to the protocol parameter-functions from Alogrithm \ref{algo:trimming} are defined as:
    $ \Delta(\P) = k' + a \log \left( S(\P, 0) + \lvert \P\{B':\} \rvert \right)$, $g(\P,\mu) = k + a\log S(\P,\mu)$, and $f(\P,\mu) = c \cdot g(\P, \mu)$. We simply state the results here without proof. The proofs can be found in Appendix \ref{app:security-proof}. All the results are with respect to the stochastic model defined in Appendix \ref{app:model_supp}.
    
    First, we show if the honest nodes have the majority of mining power, we can choose security parameters such that our protocol is secure against a trim-attack.
    
    \begin{theorem}[Security of the trimmed Chain]
            Assume an honest majority, $\lambda_h > \lambda_a$, where $\lambda_h$ and $\lambda_a$ are the mining rates of the honest and adversarial nodes respectively. Let the trimming algorithm (Algorithm \ref{algo:trimming}) parameters $k,k',c, a,\delta$ satisfy,
            \begin{equation*}
                \begin{split}
                    k' &= k - a \log\left(((1+\delta^2) + \delta)/\delta \right),  \quad
     1 < (1-\delta)^3\frac{\lambda_h}{\lambda_a}\frac{c'-1}{c'}\frac{c-1-c'}{c}, \\
  1 &< (1-\delta)^5\frac{\lambda_h}{\lambda_a} \frac{c-1}{c}, \quad
  a \geq \frac{8}{\delta^2},
                \end{split}
            \end{equation*}
            
            where $2<c'<c$ is some constant. Then\footnote{In our work, $\negl(k) = e^{-\Omega(k)} < \frac{1}{\mathsf{poly}(k)}$.}, $
     \mathbb{P}(\mathsf{trim\mhyphen attacked}) = \negl(k).$  
    \label{th:trim_attacked}
    \end{theorem}
    As a corollary to the above theorem, we can conclude the chain-selection made by hybrid-nodes is consistent with underlying chain lengths.
    
    \begin{corollary}
    The congruence property holds except with probability $\negl(k)$.  
    \label{corr:congruence}
    \end{corollary}

    \begin{theorem}[State Security]
     Assuming honest majority, $\lambda_h > \lambda_a$, and that the security parameters are as in Theorem \ref{th:trim_attacked}, then, $ \Pr(\mathsf{state\mhyphen attacked})= \negl(k)$.
     \label{th:state_attacked}
    \end{theorem}
    \begin{corollary}[Bootstrapping Security]
    Assume the arrival process of new nodes is independent of the randomness of the blockchain system. If we further assume that a new node contacts at least one honest node and that there is honest mining majority, then $\mathsf{bootstrap\mhyphen attacked}$ occurs with probability $\negl(k)$.
    \label{corr:bootstrap_attacked}
    \end{corollary}
    
    
    In the next theorem, we provide an upper bound on the cold-storage required to store the trimmed chain. This theorem relies on the assumption that a particular kind of adversary is absent. We further comment on this in Remark \ref{rm:need_optimism}. Additionally, we also note that the following theorem doesn't account for the storage required to store the state at the trimming point. We discuss optimizing the storage for the state using stateless blockchains in Section \ref{sec:stateless_blockchains}.
    
    \begin{theorem}[Optimistic Succinctness]
    Assuming that all the nodes are honest, the cold-storage requirements for a hybrid node to store the trimmed-chain $\P$ is $O(\log^{4} B)$ with high probability in $k$. Here, $B$ is the length of the underlying chain $\C$.
    \label{th:succinctness}
    \end{theorem}
    \begin{remark}[Need of Optimism for Succinctness]
   Since our protocol is an extension of NIPoPoWs, we need to rely on optimism for succinctness as well. In particular, note that when trimming we need to ensure that the level ranges are ``good''. We have also shown that under honest behaviour, the sizes of the different upchains concentrate around the mean, thus leading to good level-ranges with high probability. However, an adversary with a small mining power could launch an attack which hampers the concentrations around the means and thus makes good-superchains more unlikely. This attack is formally described in \cite{kiayias2020}. We note that, just like in NIPoPoWs, the adversary can only hurt the succinctness of our model, but not its security. Although the economics of launching an attack on the hybrid node's succinctness is yet to be studied formally, we venture to guess that such an attack wouldn't be economically viable to the adversary.
   \label{rm:need_optimism}
\end{remark}

    Clearly, the size of the untrimmed tail is a lower bound on the cold-storage required. In the next theorem we prove that the length of the chain tail has to be at least $\log B$ in order to protect a hybrid node against a $\mathsf{trim\mhyphen attack}$. This implies that our protocol is near-optimal in terms of cold storage, since it only requires $polylog(B)$ storage.
    
    \begin{theorem}[Vulnerability of Short Tails]
       Let the blockchain length be $B_t$ at time $t$, and the trimming point be $B'_t$. If at all times $t$, $\Pr\left(B_t-B'_t = o(\log B_t) \right) > \epsilon$, for some $\epsilon>0$, then there exists an adversary with small mining power, $\lambda_a < \lambda_h$, such that $\Pr(\mathsf{trim\mhyphen attacked})=1$.
       \label{th:trimming_attack}
    \end{theorem}

\section{Optimizing State Storage with Stateless Blockchains}
\label{sec:stateless_blockchains}

Our trimming protocol optimizes the amount of cold storage required to represent the PoW in the blockchain.
In this section, we outline how our work can interface with methods optimizing the storage of the \emph{state}. Furthermore, a detailed description is found in App.~\ref{sec:app-stateless_blockchains}. In hybrid nodes, states need to be stored in cold and hot storage.
Along with the length of the blockchain, the size of the blockchain-state is rapidly increasing with time.
For instance, at the time of writing Bitcoin's UTXO state is almost 4GB in size, which motivates lessening high storage and verification time by extending our work to optimize both the state's cold-storage at the trimming point, and its hot-storage at the chain tip.

We described that storing state-commitments in the blocks enables one to securely establish the state of a trimmed chain, by use of the $\mathsf{state\mhyphen verify}$ protocol described in Algorithm \ref{algo:state-verify}.
In a different line of work, called \emph{stateless blockchains}, state-commitments  are used to reduce the amount of hot-storage required for transaction validation, thus speeding up validation. We claim that when hybrid nodes are used in stateless-blockchains, then state commitments can serve a dual purpose: 1) they can be used to establish the state of the trimmed chain; 2) and, they can be used to perform stateless transaction validation.

Stateless validation, first proposed by Todd~\cite{todd}, is a scheme where nodes validating transactions store a short cryptographic state-commitment rather than the entire state. A client then provides a membership proof that the node can verify against the commitment. Recently, several constructions have been developed to perform stateless validation in both the UTXO and account-based state models using various cryptographic primitives \cite{chepurnoy2018,dryja2019,boneh2019, tomescu2020, reyzin2017, agrawal2020, gorbunov2020}.

Stateless blockchains optimize hot-storage by obviating the need to store the state at the tip of the chain in the RAM. At this point, the state at the trimming point is still stored in the cold storage. We can use this to provide an interface from trimming to stateless blockchains: First, if the proof for some client becomes outdated past the trimming point because it was offline for a long time, then the client can recompute its latest proof by querying a hybrid node for $\S(B'_t)$ and the untrimmed chain tail. Second, if a fork in the chain makes the proofs of several clients invalid, they can all recompute their proofs by contacting the hybrid node similarly.
In App.~\ref{sec:app-stateless_blockchains}, we describe a way to optimize the cold-storage of the state at the trimming point as well.

\section{Future Directions}
\label{app:future}

In this paper, we presented hybrid nodes which use trimming to optimize the cold-storage required to represent the PoW in the chain. When used in conjunction with stateless blockchains, we illustrated that hybrid nodes are optimized both in terms of cold and hot storage. In this section, we lay out some directions for future work.

We assume that there is honest majority, $\lambda_h > \lambda_a$, throughout the execution of the blockchain protocol. Given that we use novel security models, it will be interesting to study if it is economically viable for an adversary to acquire massive resources in order to attain a majority mining power temporarily, and launch an attack such as a $\mathsf{trim-attack}$. Similarly, it will be interesting to study the economics of an adversary who simply attempts to hurt the succinctness of hybrid nodes. In hurting the succinctness, fewer nodes would join as hybrid nodes. Thus, an adversary could attempt to take advantage of the reduced decentralization.

We have proved in Theorem \ref{th:trimming_attack}, that our protocol is near-optimal in terms of the storage it requires to represent a blockchain's PoW if it is to be secure against a $\mathsf{trim-attack}$. However, it is at the moment unclear if further storage optimization in terms of storing the state is possible for a  node that needs to have all the functionalities of a hybrid node (without assuming properties on the clients as we do in Section \ref{sec:stateless_blockchains}).

In our analysis, we assume that the block $\mathtt{id}$'s have a constant difficulty target throughout the blockchain's execution. It may be useful to relax this assumption in order to make our model closer to practical systems.

Finally, our work in this paper has mainly focused on theoretical analysis. An interesting line of work would be to implement the protocol and study the practical gains of using hybrid nodes.

\printbibliography

@article{sompolinsky2016bitcoin,
  title={Bitcoin's security model revisited},
  author={Sompolinsky, Yonatan and Zohar, Aviv},
  journal={arXiv preprint arXiv:1605.09193},
  year={2016}
}

@article{kadhe2019sef,
  author    = {Swanand Kadhe and
               Jichan Chung and
               Kannan Ramchandran},
  title     = {SeF: {A} Secure Fountain Architecture for Slashing Storage Costs in
               Blockchains},
  journal   = {CoRR},
  volume    = {abs/1906.12140},
  year      = {2019},
  url       = {http://arxiv.org/abs/1906.12140},
  archivePrefix = {arXiv},
  eprint    = {1906.12140},
  bibsource = {dblp computer science bibliography, https://dblp.org}
}

@article{li2021polyshard,
author = {Li, Songze and Yu, Mingchao and Yang, Chien-Sheng and Avestimehr, Amir Salman and Kannan, Sreeram and Viswanath, Pramod},
title = {PolyShard: Coded Sharding Achieves Linearly Scaling Efficiency and Security Simultaneously},
year = {2021},
issue_date = {2021},
publisher = {IEEE Press},
volume = {16},
issn = {1556-6013},
url = {https://doi.org/10.1109/TIFS.2020.3009610},
doi = {10.1109/TIFS.2020.3009610},
journal = {Trans. Info. For. Sec.},
month = jan,
pages = {249–261},
numpages = {13}
}

@INPROCEEDINGS{perard2018,
  author={D. {Perard} and J. {Lacan} and Y. {Bachy} and J. {Detchart}},
  booktitle={2018 IEEE International Conference on Internet of Things (iThings) and IEEE Green Computing and Communications (GreenCom) and IEEE Cyber, Physical and Social Computing (CPSCom) and IEEE Smart Data (SmartData)}, 
  title={Erasure Code-Based Low Storage Blockchain Node}, 
  year={2018},
  volume={},
  number={},
  pages={1622-1627},
  doi={10.1109/Cybermatics_2018.2018.00271}}

@article{raman2017,
  author    = {Ravi Kiran Raman and
               Lav R. Varshney},
  title     = {Dynamic Distributed Storage for Scaling Blockchains},
  journal   = {CoRR},
  volume    = {abs/1711.07617},
  year      = {2017},
  url       = {http://arxiv.org/abs/1711.07617},
  archivePrefix = {arXiv},
  eprint    = {1711.07617},
  bibsource = {dblp computer science bibliography, https://dblp.org}
}

@misc{bitcoinfullnode ,
author = {},
title = {Running a Full Node},
howpublished = {Available at \url{https://bitcoin.org/en/full-node} (2021/05/20)}
}

@misc{ethereumnode ,
author = {},
title = {Nodes and Clients},
howpublished = {Available at \url{https://ethereum.org/en/developers/docs/nodes-and-clients/} (2021/05/20)}
}

@misc{bitcoinhash ,
author = {},
title = {Bitcoin Hash Rate},
howpublished = {Available at \url{https://www.coinwarz.com/mining/bitcoin/hashrate-chart} (2021/04/29)}
}

@misc{ethereumhash ,
author = {},
title = {Ethereum Hash Rate},
howpublished = {Available at \url{https://www.coinwarz.com/mining/ethereum/hashrate-chart} (2021/04/29)}
}

@inproceedings{pass2017,
  title={Analysis of the blockchain protocol in asynchronous networks},
  author={Pass, Rafael and Seeman, Lior and Shelat, Abhi},
  booktitle={Annual International Conference on the Theory and Applications of Cryptographic Techniques},
  pages={643--673},
  year={2017},
  organization={Springer}
}

@article{sankagiri2021,
  title={The Longest-Chain Protocol Under Random Delays},
  author={Sankagiri, Suryanarayana and Gandlur, Shreyas and Hajek, Bruce},
  journal={arXiv preprint arXiv:2102.00973},
  year={2021}
}

@inproceedings{garay2015,
  title={The bitcoin backbone protocol: Analysis and applications},
  author={Garay, Juan and Kiayias, Aggelos and Leonardos, Nikos},
  booktitle={Annual international conference on the theory and applications of cryptographic techniques},
  pages={281--310},
  year={2015},
  organization={Springer}
}

@inproceedings{kiayias2016,
  title={Proofs of proofs of work with sublinear complexity},
  author={Kiayias, Aggelos and Lamprou, Nikolaos and Stouka, Aikaterini-Panagiota},
  booktitle={International Conference on Financial Cryptography and Data Security},
  pages={61--78},
  year={2016},
  organization={Springer}
}

@inproceedings{kiayias2020,
  title={Non-interactive proofs of proof-of-work},
  author={Kiayias, Aggelos and Miller, Andrew and Zindros, Dionysis},
  booktitle={International Conference on Financial Cryptography and Data Security},
  pages={505--522},
  year={2020},
  organization={Springer}
}

@inproceedings{eyal2014,
  title={Majority is not enough: Bitcoin mining is vulnerable},
  author={Eyal, Ittay and Sirer, Emin G{\"u}n},
  booktitle={International conference on financial cryptography and data security},
  pages={436--454},
  year={2014},
  organization={Springer}
}

@inproceedings{dembo2020,
  title={Everything is a race and nakamoto always wins},
  author={Dembo, Amir and Kannan, Sreeram and Tas, Ertem Nusret and Tse, David and Viswanath, Pramod and Wang, Xuechao and Zeitouni, Ofer},
  booktitle={Proceedings of the 2020 ACM SIGSAC Conference on Computer and Communications Security},
  pages={859--878},
  year={2020}
}

@inproceedings{gavzi2020,
  title={Tight consistency bounds for bitcoin},
  author={Ga{\v{z}}i, Peter and Kiayias, Aggelos and Russell, Alexander},
  booktitle={Proceedings of the 2020 ACM SIGSAC Conference on Computer and Communications Security},
  pages={819--838},
  year={2020}
}

@inproceedings{matzutt2020,
  title={How to Securely Prune Bitcoin’s Blockchain},
  author={Matzutt, Roman and Kalde, Benedikt and Pennekamp, Jan and Drichel, Arthur and Henze, Martin and Wehrle, Klaus},
  booktitle={2020 IFIP Networking Conference (Networking)},
  pages={298--306},
  year={2020},
  organization={IEEE}
}

@misc{poissontail ,
author = {Cl\'ement Canonne},
title = {A Short Note on Poisson Tail Bounds},
howpublished = {Available at \url{http://www.cs.columbia.edu/~ccanonne/files/misc/2017-poissonconcentration.pdf} (2021/05/18)}
}

@misc{todd ,
author = {Peter Todd},
title = {Making UTXO Set Growth Irrelevant With Low-Latency Delayed TXO Commitments},
howpublished = {Available at \url{https://petertodd.org/2016/delayed-txo-commitments} (2021/05/24)}
}

@article{dubhashi1998,
  title={Concentration of measure for the analysis of randomised algorithms},
  author={Dubhashi, D and Panconesi, Alessandro},
  journal={Draft Manuscript, http://www. brics. dk/ale/papers. html},
  year={1998}
}

@inproceedings{reddy2021,
  title={securePrune: Secure block pruning in UTXO based blockchains using Accumulators},
  author={Reddy, B Swaroopa},
  booktitle={2021 International Conference on COMmunication Systems \& NETworkS (COMSNETS)},
  pages={174--178},
  year={2021},
  organization={IEEE}
}

@inproceedings{Daveas2020,
author = {Daveas, Stelios and Karantias, Kostis and Kiayias, Aggelos and Zindros, Dionysis},
title = {A Gas-Efficient Superlight Bitcoin Client in Solidity},
year = {2020},
isbn = {9781450381390},
publisher = {Association for Computing Machinery},
address = {New York, NY, USA},
url = {https://doi.org/10.1145/3419614.3423255},
doi = {10.1145/3419614.3423255},
booktitle = {Proceedings of the 2nd ACM Conference on Advances in Financial Technologies},
pages = {132–144},
numpages = {13},
location = {New York, NY, USA},
series = {AFT '20}
}

@inproceedings{valiant2008incrementally,
  title={Incrementally verifiable computation or proofs of knowledge imply time/space efficiency},
  author={Valiant, Paul},
  booktitle={Theory of Cryptography Conference},
  pages={1--18},
  year={2008},
  organization={Springer}
}

@article{kattis2020proof,
  title={Proof of Necessary Work: Succinct State Verification with Fairness Guarantees.},
  author={Kattis, Assimakis and Bonneau, Joseph},
  journal={IACR Cryptol. ePrint Arch.},
  volume={2020},
  pages={190},
  year={2020}
}

@article{bonneau2020coda,
  title={Coda: Decentralized Cryptocurrency at Scale.},
  author={Bonneau, Joseph and Meckler, Izaak and Rao, Vanishree and Shapiro, Evan},
  journal={IACR Cryptol. ePrint Arch.},
  volume={2020},
  pages={352},
  year={2020}
}

@techreport{chen2020reducing,
  title={Reducing Participation Costs via Incremental Verification for Ledger Systems},
  author={Chen, Weikeng and Chiesa, Alessandro and Dauterman, Emma and Ward, Nicholas P},
  year={2020},
  institution={Cryptology ePrint Archive, Report 2020/1522}
}

@article{dryja2019,
  title={Utreexo: A dynamic hash-based accumulator optimized for the Bitcoin UTXO set.},
  author={Dryja, Thaddeus},
  journal={IACR Cryptol. ePrint Arch.},
  volume={2019},
  pages={611},
  year={2019}
}

@article{chepurnoy2018,
  title={Edrax: A Cryptocurrency with Stateless Transaction Validation.},
  author={Chepurnoy, Alexander and Papamanthou, Charalampos and Zhang, Yupeng},
  journal={IACR Cryptol. ePrint Arch.},
  volume={2018},
  pages={968},
  year={2018}
}

@inproceedings{boneh2019,
  title={Batching techniques for accumulators with applications to iops and stateless blockchains},
  author={Boneh, Dan and B{\"u}nz, Benedikt and Fisch, Ben},
  booktitle={Annual International Cryptology Conference},
  pages={561--586},
  year={2019},
  organization={Springer}
}

@inproceedings{tomescu2020,
  title={Aggregatable subvector commitments for stateless cryptocurrencies},
  author={Tomescu, Alin and Abraham, Ittai and Buterin, Vitalik and Drake, Justin and Feist, Dankrad and Khovratovich, Dmitry},
  booktitle={International Conference on Security and Cryptography for Networks},
  pages={45--64},
  year={2020},
  organization={Springer}
}

@inproceedings{reyzin2017,
  title={Improving authenticated dynamic dictionaries, with applications to cryptocurrencies},
  author={Reyzin, Leonid and Meshkov, Dmitry and Chepurnoy, Alexander and Ivanov, Sasha},
  booktitle={International Conference on Financial Cryptography and Data Security},
  pages={376--392},
  year={2017},
  organization={Springer}
}

@inproceedings{agrawal2020,
  title={KVaC: Key-Value Commitments for Blockchains and Beyond},
  author={Agrawal, Shashank and Raghuraman, Srinivasan},
  booktitle={International Conference on the Theory and Application of Cryptology and Information Security},
  pages={839--869},
  year={2020},
  organization={Springer}
}

@inproceedings{gorbunov2020,
  title={Pointproofs: aggregating proofs for multiple vector commitments},
  author={Gorbunov, Sergey and Reyzin, Leonid and Wee, Hoeteck and Zhang, Zhenfei},
  booktitle={Proceedings of the 2020 ACM SIGSAC Conference on Computer and Communications Security},
  pages={2007--2023},
  year={2020}
}

@inproceedings{catalano2013vector,
  title={Vector commitments and their applications},
  author={Catalano, Dario and Fiore, Dario},
  booktitle={International Workshop on Public Key Cryptography},
  pages={55--72},
  year={2013},
  organization={Springer}
}

@inproceedings{bunz2020,
  title={Flyclient: Super-light clients for cryptocurrencies},
  author={B{\"u}nz, Benedikt and Kiffer, Lucianna and Luu, Loi and Zamani, Mahdi},
  booktitle={2020 IEEE Symposium on Security and Privacy (SP)},
  pages={928--946},
  year={2020},
  organization={IEEE}
}

@incollection{karantias2020,
  title={Compact storage of superblocks for nipopow applications},
  author={Karantias, Kostis and Kiayias, Aggelos and Zindros, Dionysis},
  booktitle={Mathematical Research for Blockchain Economy},
  pages={77--91},
  year={2020},
  publisher={Springer}
}

@article{kiayias2020velvet,
  title={The Velvet Path to Superlight Blockchain Clients},
  author={Kiayias, Aggelos and Polydouri, Andrianna and Zindros, Dionysis},
  year={2020}
}

@misc{Kiayias2021Mining,
    author       = {Aggelos Kiayias and
		    Nikos Leonardos and
		    Dionysis Zindros},
    title        = {Mining in Logarithmic Space},
    howpublished = {Cryptology ePrint Archive, Report 2021/623},
    year         = {2021},
    note         = {\url{https://ia.cr/2021/623}},
}

@misc{bitcoinutxo ,
author = {},
title = {Bitcoin UTXO size},
howpublished = {Available at \url{https://tinyurl.com/cr7w2ep5} (2021/09/13)}
}

\newpage

\appendix

\section{Other Related Works}
\label{app:other_work}

\paragraph{Coding Based Storage Reduction}
Several recent works have employed erasure-codes to reduce cold storage costs in blockchains \cite{perard2018, li2021polyshard, kadhe2019sef, raman2017}. The core idea in these works is that each node stores a coded version of a random subset of the blockchain. A new joining node can then download the coded fragments from several nodes and reconstruct the entire blockchain. Although this does lead to a significant reduction in storage requirements, the storage in these methods has an information theoretic lower bound of $\Omega(B/N)$, where $B$ is the length of the blockchain and $N$ is the number of nodes \cite{li2021polyshard}. In our proposed method, the cold storage cost scales \emph{poly-logarithmically} in the length of the blockchain.

\paragraph{Incrementally Verifiable Computation} Incrementally verifiable computation (IVC) is an orthogonal approach, initially proposed for general settings by Valiant \cite{valiant2008incrementally}. IVC's provide a succinct argument of knowledge that some state is correctly computed from the initial state of computation (and all preceding states of computation) by performing a recursive composition of state transitions.
In the context of blockchains, \cite{kattis2020proof,bonneau2020coda, chen2020reducing} employ IVC to reduce state verification time.
For the sake of comparison, we note that these methods are based on generic techniques, like applying a Succinct Non-interactive ARgument of Knowledge (SNARK) on a suitable compliance predicate and recursively composing such SNARKs, which result in concretely high overhead costs for the prover. Furthermore, in \cite{kattis2020proof}, which is most directly related to our work amongst the three, the authors acknowledge that their approach does not lessen storage requirements of mining nodes. In contrast, the hybrid nodes which have mining capability and more.

\paragraph{Lightweight Clients with Sublinear Complexity} Recall lightweight clients only seek to verify payment proofs. Apart from NIPoPoW, another popular method that has sublinear complexity for lightweight clients is \emph{FlyClient} \cite{bunz2020}. It uses probabilistic block sampling and a novel commitment scheme called merkle-mountain-ranges. FlyClient is more versatile that NIPoPoW's in that it doesn't require optimism for succinctness, and it allows for varying block difficulties. Unfortunately, at the moment, we don't believe the FlyClient protocol can be used for hybrid nodes. In addition to theoretical works, there have been recent works that have also studied practical aspects of NIPoPoW and FlyClient \cite{karantias2020, Daveas2020, kiayias2020velvet}.

\section{Additional Model Details}
\label{app:model_supp}

\subsection{Communication Model}
Several communication models are considered in the literature. The simplest is the \emph{synchronous} model where a block broadcast by a node at a certain time is received by all other nodes immediately \cite{garay2015, eyal2014}. Since time is continuous, no more than one block is mined at any given time, implying only one block could be in communication at any given time. 
More complicated communication models with communication delays are also considered in the literature \cite{pass2017, sankagiri2021}. For the sake of simplicity, we consider the synchronous model in this paper and leave it to future work to transfer our results to more complicated communication models.

\subsection{Additional Notation}

We define the notion of a \emph{downchain} slightly different from \cite{kiayias2016}. Let $\C' \subseteq \C\uparrow^{\mu}$ be a level-$\mu$ superchain. Let $\C'[0] = \C\uparrow^{\mu}[j]$. Then, the corresponding downchain is defined as, 

\begin{equation}\label{eq:downchain}
    \C'\downarrow \triangleq \begin{cases}
        \C[\C\uparrow^{\mu}[j-1]+1:\C'[-1]], & \text{if $j>0$}, \\
        \C[:\C'[-1]], & \text{if $j=0$}.
    \end{cases}
\end{equation}

\subsection{Stochastic Model for the Blockchain. }
\label{ssec:stochastic_model}
In order to study the behavior of the blockchain, we consider a stochastic model for the blockchain which has been used previously in the literature to study blockchain security \cite{dembo2020,gavzi2020, sompolinsky2016bitcoin}. The stochastic model allows us to abstract away the cryptographic details of PoW mining and focus specifically on the dynamics of the blockchain growth. In particular, we model PoW mining as calls to an ideal random functionality. At each query, independent of everything else, the ideal functionality returns ``success'' with a small probability $p$ and ``failure'' otherwise. After describing our protocol and its analysis, we interpret our results under popular cryptographic models such as the model used by Pass., et al \cite{pass2017} and the Bitcoin Backbone Model \cite{kiayias2016} in Section \ref{ssec:pass}. 

Since we use a continuous time model, we represent mining as a Poisson process as has been used in previous works \cite{gavzi2020, dembo2020}. Let $\Omega$ be the sample-space, and let $\omega \in \Omega$ denote an outcome (also called fixed randomness). Let the overall mining process of honest nodes be a Poisson process of rate $\lambda_h$ and let the adversary's mining be a Poisson process of rate $\lambda_a$. The honest and adversarial processes are independent of each other. Honest blocks add their mined blocks to the tip of their longest chains and immediately broadcast the block. However, the adversary could add their mined block at any point on the chain. Moreover, the adversary may choose to keep their mined blocks secret and broadcast them at any future time of their choice. Once a block is broadcast, all the honest nodes receive it immediately. Honest nodes are said to have a \emph{majority} if $\lambda_h > \lambda_a$.

Observe that a level-$\mu$ superblock is $2^\mu$ times harder to find than a regular block. In other words, a regular block is a level-$\mu$ superblock with probability $2^{-\mu}$. Therefore, the arrival process for honest (adversarial) level-$\mu$ superblocks is a Poisson process of rate $\lambda_h/2^{\mu}$ \hspace{.2cm}( $\lambda_a/2^{\mu}$). Here as well, the honest and adversarial arrival processes are independent of each other. However, we note that the arrival processes for honest (adversarial) level-$\mu_1$ and level-$\mu_2$ superblocks are not independent even if $\mu_1 \neq \mu_2$. This is because a level-$\mu_1$ superblock is also a level-$\mu_2$ superblock if $\mu_1 \geq \mu_2$.

Next, we also capture the state-commitment process using the ideal random functionality. Recall that we include the Merkle-root $y(\S)$ of the $\S$ at that point in the block header. An adversary could attempt to produce a malicious state, $\S' \neq \S$, which also produces the same Merkle-root $y(\S)$. During this attempt, the adversary would query the tuple $(\S', y(\S))$ to the random functionality and receive ``success'' with a very small probability $q$ and ``failure'' otherwise, independently of everything else. If the adversary receives a success, we say $\S' \equiv \S$. Here, $q$ can be assumed to be very small compared to the block success probability $p$. This is because a block requires only the first $T$ bits of the hash-string to be $0$, whereas in the case of the malicious state, the entire hash-string has to match the original commitment $y$. In other words, the adversary must find a \emph{collision} on the hash function used in the Merkle tree, which is assumed to occur only with negligible probability with respect to $p$.
Assume that the adversary tries to find a malicious state at the tip of the chain at some time $t$. By time $s\geq t$, let $N_t(s)$ denote the number of subsequent blocks' states for which the adversary's malicious state sequence is equivalent to, i.e.,
\begin{equation*}
    F^n(\S'(\C_t)) \equiv F^n(\S(\C_t)) \quad ,\text{for all } 0 \leq n < N_t(s).
\end{equation*}
Then, $N_t(s)$ is a Poisson process of rate $\lambda_s$, where $\lambda_s << \lambda_a$ because $q << p$. 

\section{Our Protocols: Further Details}
\label{app:protocol}

The notion of a superchain's weight being approximately equal to the PoW of the corresponding downchain was formalized as \emph{superquality} in \cite{kiayias2020}. Adding to this idea, we define a notion of a dominant superchain below.

\subsubsection*{Good Superchain} First, we require that every sufficiently large suffix of the upchain weighs at least $(1-\delta)$ times the amount of PoW of its corresponding downchain.

\begin{definition}[Superquality]
Let $\C$ be a blockchain, and $\C\uparrow^\mu$ be its level-$\mu$ upchain. $C\uparrow^{\mu}$ is said to have $(\delta, g)$ superquality, for $\delta > 0$ and $g \in \mathbb{N}$, if for all $g' \geq g$, the following is true:
\begin{align*}
    \lvert \C\uparrow^{\mu}[-g':]\rvert
    \geq (1-\delta)2^{-\mu} \lvert \C\uparrow^{\mu}[-g':]\downarrow \rvert.
\end{align*}\label{def:superchain_quality}
\end{definition}
Next, we define a dominant superchain.

    \begin{definition}[Dominant Superchain]
Consider a \\
$\mu$-superchain $\C'$ of an underlying chain, $\C = \C'\downarrow$, and parameter $g \in \mathbb{N}$. $\C'$ is said to be a dominant superchain with parameter $g$ if the following is true. Let $\P' \subseteq \C$ be any trimmed chain  with associated level range functions $L'_f$ and $L'_l$. Let, $\mu_1 \leq \mu' \leq \mu$, and $ L'_f(\mu') \leq l \leq L'_l(\mu')$. Then,
\begin{equation}
\begin{split}
    2^{\mu_1}\lvert \P'\{l:&L'_l(\mu')+1\}\uparrow^{\mu_1} \rvert \\
    + &\sum_{\mu'' < \mu'} 2^{\mu''}\lvert \P'\{L'_f(\mu''):L'_l(\mu'')+1]\uparrow^{\mu''} \rvert \geq 2^{\mu}g,
    \end{split}
    \label{eq:dom_supchain_2}
\end{equation}
implies that,
\begin{equation*}
    \lvert \P'\uparrow^{\mu} \rvert \geq 1.
\end{equation*}
\label{def:dominant_superchain}
\end{definition}

As intuition, first note that $\C'$ and $\P'$ have the same underlying downchain. We are interested in the weight of $\P'$ to the right of a block $l$ in level range $\mu'$. For reasons that will be clear from our trimming and chain-compare protocols (Algorithms \ref{algo:trimming} and \ref{algo:compare}), we allow weighing the first level range $\mu'$ using a different level $\mu_1 \leq \mu'$. Then, the dominant-superchain property says that if $\P'$ has a weight that is large enough, then it contains at least 1 level-$\mu$ superblock.

A superchain is a \emph{good-superchain} if it has superchain quality and is a dominant superchain.

\begin{definition}[Good Superchain]
A $\mu$-superchain $\C'$ of an underlying chain $\C = \C'\downarrow$, is a good-superchain, denoted $\mathtt{good}_{\delta, g}(\C', \mu)$, if it satisfies superquality with parameters $(\delta, g)$ and is a dominant superchain with parameter $g$.
\end{definition}

Intuitively, if all the respective quantities above are equal to their mean values, then $\C\uparrow^{\mu}$ is a good superchain. As long as an adversary does not try to violate the concentration around the means, a superchain is a good superchain with high probability. The proof for superquality can be found in \cite{kiayias2020}. The proof for dominant superchain property can be found in Appendix \ref{app:dominant_superchain}.
    
\subsection{Trimming} 
With the notion of a good superchain defined, our trimming protocol description from Section \ref{sec:trimming} can be completed by specifying the parameter functions used in Algorithm \ref{algo:trimming}. We next describe our trimming protocol. The protocol has security parameters $k, k' \in \mathbb{N}$, $a, c \in \mathbb{R}^+$, and $\delta \in (0,1)$. The length of the untrimmed tail that is left untouched is dictated by the function $\Delta(\P)$. If we denote $B'$ to be the old trimming point, then let\footnote{Unless otherwise mentioned, logarithms are to the base $e$.},
\begin{equation}
    \Delta(\P) = k' + a \log \left( S(\P, 0) + \lvert \P\{B':\} \rvert \right).
    \label{eq:Delta}
\end{equation}

When $\P\{:B'\}$ is trimmed to a certain level $\mu$, we first ensure that the relevant upchain is a good superchain. The parameter $g$ used here for level $\mu$ is given by, 
\begin{equation}
        g(\P,\mu) = k + a\log S(\P,\mu).
        \label{eq:g}
\end{equation}

Apart from requiring that the $\mu$-upchain is a good-superchain, we require that there are ``enough" number of superblocks of that level $\mu$. This is given by
\begin{equation}
        f(\P,\mu) = c \cdot g(\P, \mu)
        \label{eq:f}
\end{equation}

We require the functions $\Delta, g$ and $f$ to depend on the logarithm of the estimated PoW to obtain strong concentration bounds, which achieves security properties. As far as we know, this notion was first used to obtain strong security for blockchains in \cite{sompolinsky2016bitcoin} where the authors call it ``logarithmic waiting time''.

 \subsection{Comparing Work in trimmed Chains}
\label{ssec:compare}
Hybrid nodes need to have a protocol, $\mathsf{Compare}(\C^{(1)}, \C^{(2)})$, to chose the main chain given two conflicting chains, $\C^{(1)}$ and $\C^{(2)}$. Full nodes (that store the entire blockchain) simply choose the longer of the two chains as the main chain. The chain selection protocol for hybrid nodes is a little more complicated since they do not store the entire chain. At a high level, they use the sum of the cumulative weight, $S(\P,0)$, of the trimmed portion of a chain and the length of untrimmed section of the chain as a proxy for the chain length. Next, we explain this algorithm more formally.

Let $\P^{(1)}$ and $\P^{(2)}$ be two trimmed chains with underlying (complete) chains $\C^{(1)}$ and $\C^{(2)}$. We need a protocol, $\mathsf{Compare}(\P^{(1)}, \P^{(2)})$, that outputs $\P^{(1)}$ if $\lvert \C^{(1)} \rvert > \lvert \C^{(2)} \rvert $, otherwise it outputs $\P^{(2)}$. We allow error with small probability given by the security definitions in Section \ref{sec:security-statements}.
    
The protocol is shown in Algorithm \ref{algo:compare}  and we summarize it here. First, only the parts of the chains after their LCA block need to be compared (line 1). The algorithm assigns weights to each chain, denoted $\mathsf{Weight}(\P)$, and declares the chain with the higher weight the winner (line 5). The weight is assigned to (the non-common part of) a trimmed chain as follows. If the chain doesn't have a long enough untrimmed tail, then its weight is set to 0 (line 8). If the LCA block $b$ is after the trimming point, then the weight is set to the number of blocks after $b$ (lines 11 and 12). Otherwise, the level range in which $b$ lies is computed (line 14). For the first level range, a set $M$ of levels is formed with enough number of superblocks of that level (line 16). Then, the highest PoW among the upchains in set $M$ is considered (line 17). For all subsequent level-ranges, the PoW corresponding their level is added (lines 18 to 20). At the end, the PoW of the untrimmed-tail is added and returned (line 22).

    \begin{algorithm}[t]
        \DontPrintSemicolon
        \SetKwInOut{Input}{input~}
        \SetKwInOut{Local}{Local variables~}
        \SetKwInOut{Output}{output~}
        \DontPrintSemicolon
        
        \Input{
            $\P^{(1)}, L_f^{(1)}, L_l^{(1)}, {B'}^{(1)}$ :: Chain 1\\
            $\P^{(2)}, L_f^{(2)}, L_l^{(2)}, {B'}^{(2)}$ :: Chain 2\\
        }
        
        \vspace{\baselineskip}
          
        \SetKwFunction{FWeight}{Weight}
        \SetKwProg{Fn}{func}{:}{}
        
        \SetKwProg{On}{on event}{:}{}

            $b = \mathsf{LCA}(\P^{(1)},\P^{(2)})$ \;
            $W^{(1)} = \FWeight(\P^{(1)}, L_f^{(1)}, L_l^{(1)}, {B'}^{(1)}, b)$ \;
            $W^{(2)} = \FWeight(\P^{(2)}, L_f^{(2)}, L_l^{(2)}, {B'}^{(2)}, b)$ \;
            
            \;
            \KwRet argmax$(W^{(1)}, W^{(2)})$ \;
        
        \;
        
        \Fn{\FWeight{$\P$, $L_f$, $L_l$ $B'$, $b$}}
        {
        
            \If{$\lvert \P\{B':\}  \rvert < \Delta(g(\P, 0)+ \lvert \P\{B':\}  \rvert )$}
            {
                \KwRet $0$ \;
            }
        
            \;
            
            \If{$b \geq B'$}
            {
                \KwRet $\lvert \P\{b:\} \rvert$ \;
            }
        
            \;
        
            $\hat{\mu} \leftarrow$ such that $L_f(\hat{\mu}) \leq b \leq L_l(\hat{\mu})$ \;
            \;

            $M =  \Big\{\mu': \lvert \P\uparrow^{\mu'}\{ b: L_l(\hat{\mu})+1 \} \rvert \geq f(\P, \hat{\mu}) \Big\}$ \;
            $W = \max_{\mu' \in M \cup \{0\}} \Big\{2^{\mu'} \lvert \P\uparrow^{\mu'} \{b:L_l(\hat{\mu})+1\} \rvert \Big\}.$ \;
                
            \For{$\mu$ from $\hat{\mu}-1$ to $0$}
            {
                \If{$\lvert \P\uparrow^{\mu} \{L_f(\mu):L_l(\mu)+1\} \rvert \geq f(\P,\mu)$}
                {
                $W = W + 2^{\mu} \lvert \P\uparrow^{\mu} \{L_f(\mu):L_l(\mu)+1\} \rvert .$ \;
                }
            }
            \;
            \KwRet $W + \lvert \P\{B':\}  \rvert$ \;
        }
            
    \caption{$\mathsf{Compare}(\cdot, \cdot)$ Protocol to compare work in trimmed chains}
    \label{algo:compare}
    \end{algorithm} 
    
    The $\mathsf{Compare}(\cdot,\cdot)$ algorithm is used in two ways. First, by hybrid-nodes to compute the main chain when they receive another (trimmed) chain which is not the same as their trimmed chain. Second, by new hybrid-nodes to bootstrap into the system. The new node receives (trimmed) chains from other nodes, and then calls  $\mathsf{Compare}(\cdot, \cdot)$ on pairs of chains, until it is left with 1 winner.
    
    
    
    
    
    
    

 \subsection{State Verification}
\label{ssec:state_verify}
Similar to CoinPrune \cite{matzutt2020}, a short commitment to the state at the block is stored in
every block. Therefore, a hybrid node can verify the correctness of the blockchain's state by
comparing it to the corresponding state commitment. The protocol is called $\mathsf{State-Verify}$
and is presented with more details in Algorithm \ref{algo:state-verify}, and it is summarized in the
following paragraph.

After a new node adopts a trimmed chain $\P$ with trimming point $B'$, it needs to download the state, $\S(B')$, at its trimming point. Since all blocks to the right of the trimming point are retained, the node can compute all the states corresponding to blocks in the untrimmed tail. In order to verify the state, the node checks whether the Merkle-root of the computed state for each block in the untrimmed tail matches its corresponding Merkle-root $y$ stored within the block-header. If the Merkle-roots do not match at any block in the untrimmed tail, then the algorithm returns False. This is shown in Algorithm \ref{algo:state-verify}.

\begin{algorithm}[H]
        \DontPrintSemicolon
        \SetKwInOut{Input}{input~}
        \SetKwInOut{Local}{Local variables~}
        \SetKwInOut{Output}{output~}
        \DontPrintSemicolon
        
        \Input{
            $\P$ :: trimmed chain\\
            $B' ::$ trimming point\\
            $\S(B')$ :: State at trimming point\\
        }
        
        \vspace{\baselineskip}
          
        
        \SetKwProg{On}{on event}{:}{}
        
            \For{$i$ from $0$ to $B-B'$}
            {
                Compute $\S(B'+i)$ \;
                \If{$\mathsf{Merkle-root}(\S(B'+i)) \neq y(B'+i)$}
                {
                    \KwRet False \;
                }
        }
        \KwRet True \;
            
    \caption{$\mathsf{State-Verify}$ Protocol to verify the state of a trimmed chain}
    \label{algo:state-verify}
    \end{algorithm}

\section{Poisson Process Preliminaries}
\label{app:poisson_properties}

$X$ is a Poisson random variable with rate $\lambda$, denoted $X \sim Poisson(\lambda)$, if its probability distribution satisfies,
\begin{equation*}
    \P(X = x) = e^{-\lambda} \frac{\lambda^x}{x!}, \quad \text{for } x \in \mathbb{Z}^+.
\end{equation*}

Let $A(t)$ be a random process indexed by time $t \geq 0$. Denote $A(t_1,t_2) = A(t_2)-A(t_1)$, where $t_2 \geq t_1$. $A$ is a Poisson point process (or simply, Poisson process) of rate $\lambda$, denoted as $A \sim PPP(\lambda)$, if,
\begin{equation*}
   A(t) \sim Poisson(\lambda t), \text{for all } t \geq 0, 
\end{equation*}
and for disjoint intervals $(t_1,t_2)$ and $(t_3, t_4)$,
\begin{equation*}
   A(t_1,t_2) \text{ is independent of } A(t_3, t_4). 
\end{equation*}

Since the range of $A$ is the positive integers, and can be shown to be non-decreasing with time, it can be thought of as a counting process. We also note the fact the time between two successive events in $A$ is called the interarrival time, and it is distributed as an exponential random variable of rate $\lambda$.

First, we recall concentration bounds for a Poisson random variable, without proof. The proof is obtained by applying Chernoff Bound \cite{poissontail}.
\begin{lemma}
 Let $X \sim Poisson(\lambda)$, for some $\lambda > 0$. Then, for any $x>0$,
 \begin{align*}
     \Pr(X \geq \lambda + x) &\leq e^{-\frac{-x^2}{2(\lambda + x)}}, \\
      \Pr(X \leq \lambda - x) &\leq e^{-\frac{-x^2}{2(\lambda + x)}}.
 \end{align*}
 \label{lm:poissontail}
\end{lemma}

In the following lemma, we upper bound and lower the time required for $n$ events to occur in a Poisson process.

\begin{lemma}
Let $A \sim PPP(\lambda)$. Let $T$ be a random variable indicating the time taken for $n$ events to occur in the Poisson process. That is, $A(T)=n$. Then, except with probability $e^{-n \frac{{\delta'}^2}{2}}$, $T \leq \frac{n(1+\delta')}{\lambda}$. Similarly, except with probability $e^{-n \frac{{\delta'}^2}{2}}$, $T \geq \frac{n(1-\delta')}{\lambda}$. 
\label{lm:time_bound}
\end{lemma}
\begin{proof}
Recall that in Poisson point process, interarrival times are i.i.d., exponential. Therefore, for $X_1, \dots, X_n \stackrel{iid}{\sim} exp(\lambda_1)$, $T \stackrel{d}{=} \sum_{i=1}^n X_i$. We can use Chernoff bound on $T$ as follows:

\begin{align*}
    \Pr\left( T \geq \frac{n(1+\delta')}{\lambda_1} \right)
         &\leq e^{-n \sup_{s > 0} \left( \frac{s(1+\delta')}{\lambda_1} + \log (1-s/\lambda_1) \right)}, \\
         &= e^{-n(\delta' - \log(1+\delta'))}, \\
         &\leq e^{-n \frac{{\delta'}^2}{2}}.
\end{align*}
The upper bound is obtained similarly.
\end{proof}

\begin{lemma}
Let $A_1 \sim PPP(\lambda_1)$ and $A_2 \sim PPP(\lambda_2)$ be independent Poisson point processes, where $\lambda_2 < \lambda_1$. 

Assume that $A_1$ observers $n$ events in a certain time range. Then, $A_2$ observes at most $(1+\delta')^2\lambda_2/\lambda_1 n$ events except with probability $2 e^{- \frac{n \lambda_2}{2 \lambda_1} {\delta'}^2}$.

Similarly, assuming that $A_2$ observers $n$ events, $A_1$ observes at most $(1+\delta')^2\lambda_1/\lambda_2 n$ events except with probability $2 e^{- \frac{n}{2} {\delta'}^2}$.
\label{lm:poisson_length_compare}
\end{lemma}
\begin{proof}

From Lemma \ref{lm:time_bound}, the time required from process $A_1$ to observe $n$ events is upper bounded by $n(1+\delta')/\lambda_1$ except with probability $e^{-n{\delta'}^2/2}$. So, over this range of time, except with probability $e^{-n \frac{{\delta'}^2}{2}}$, $A_2$ is stochastically upper bounded by $Y \sim Poisson\left( n \frac{\lambda_2(1+\delta')}{\lambda_1} \right)$. Now, applying the bound for a Poisson random variable from Lemma \ref{lm:poissontail},

\begin{align*}
    \Pr\left( Y \geq \frac{n \lambda_2}{\lambda_1} (1+\delta')^2 \right) 
    &\leq \exp \left( - \frac{n \left( \frac{\lambda_2(1+\delta')}{\lambda_1} \right)^2 {\delta'}^2}{2 n \frac{\lambda_2}{\lambda_1}(1+\delta')^2} \right), \\
    &= \exp \left( - \frac{n \lambda_2}{2 \lambda_1} {\delta'}^2 \right).
\end{align*}

Therefore, over the time when $A_1$ observes $n$ events, $A_2 \leq (1+\delta')^{2}\lambda_2/\lambda_1 n$ except with probability $2 e^{- \frac{n \lambda_2}{2 \lambda_1} {\delta'}^{2}}$.

The other bound is obtained similarly.
\end{proof}

\begin{lemma}
Let $A_1 \sim PPP(\lambda_1)$ and $A_2 \sim PPP(\lambda_2)$ be independent Poisson point processes, and $A = A_1 + A_2$. Given that over some range of time $A = n$, then $A_1 \leq \frac{n \lambda_1}{\lambda_1 + \lambda_2}(1+\delta')$ except with probability $e^{-\frac{n\lambda_1}{3(\lambda_1 + \lambda_2)}\frac{{\delta'}^{(2)}}{3}}$.
\end{lemma}
\begin{proof}
Given $A$, $A_1$ behaves like a binomial random variable: $(A_1 \big\vert A = n) \sim \mathtt{Bin}\left( n, \frac{\lambda_1}{\lambda_1 + \lambda_2} \right)$. Therefore, the upper bound on $A_1$ is obtained by using the standard Chernoff bound for binomial random variables.
\end{proof}


\section{Security Proofs}
\label{app:security-proof}

Let's recall the notation. $\P^{(1)}$ is the honest trimmed chain and $\P^{(2)}$ is the adversarial chain. The underlying chains be $\C^{(1)}$ and $\C^{(2)}$ respectively. Their respective trimming points are ${B'}^{(1)}(t)$ and ${B'}^{(2)}(t)$, and their respective chain lengths are $B^{(1)}$ and $B^{(2)}$. Let $b = \mathsf{LCA}(\P^{(1)}, \P^{(2)})$. And, $b' = \mathsf{LCA}(\C^{(1)}, \C^{(2)})$. Trivially, $b' \geq b$. Observe that $b'$ and $b$ may not be equal because the trimmed chains may skip blocks and thus have an earlier LCA block.

$L^{(1)}(\cdot)$ is the level range function for $\P^{(1)}$, and $L^{(2)}$ for $\P^{(2)}$. Denote,

\begin{align*}
    \mu^{(1)} \text{ such that } L^{(1)}(\mu^{(1)}+1) < b \leq L^{(1)}(\mu^{(1)}), \\
    \mu^{(2)} \text{ such that } L^{(2)}(\mu^{(2)}+1) < b \leq L^{(2)}(\mu^{(2)}).
\end{align*}

From the method by which $\mathsf{Weight}()$ function in Algorithm \ref{algo:compare} chooses levels, for the level ranges of $\mu^{(1)}$ and $\mu^{(2)}$, the $\mathsf{Weight}()$ function might use levels $\hat{\mu}^{(1)}\leq \mu^{(1)}$ and $\hat{\mu}^{(2)}\leq \mu^{(2)}$ respectively (see line 16 and 17 of Algorithm \ref{algo:compare}). For subsequent level ranges, the corresponding levels are used by $\mathsf{Weight()}$ (see line 20).

First, we bound the gap between the LCA of the pruned chains, $b$, and the LCA of the underlying chains, $b'$.

\begin{lemma}\label{lm:bounding_adversary_work}
Denote, 
\begin{align*}
    M = \max &\Bigg\{\; 2^{\hat{\mu}^{(1)}} g\Big(\P^{(1)}, \mu^{(1)} \Big), \quad
    2^{\mu^{(1)}-1} g\Big(\P^{(1)},\mu^{(1)}-1\Big)   \;\Bigg\}.
\end{align*}
$M$ essentially indicates the maximum of PoW of the first and second level ranges of the honest trimmed chain.

Let $b'$ be in the ${\mu'}^{(2)}$-level range in $\P^{(2)}$. That is,
\begin{equation*}
    L^{(2)}({\mu'}^{(2)}+1) < b' \leq L^{(2)}({\mu'}^{(2)}).
\end{equation*}
If ${\mu'}^{(2)} = \mu^{(2)}$, then, 
\begin{equation*}
    2^{\hat{\mu}^{(2)}}\lvert \P^{(2)}\{b+1:b'+1\}\uparrow^{\hat{\mu}^{(2)}} \rvert \leq M.
\end{equation*}
Otherwise, 
\begin{align*}
    2^{\hat{\mu}^{(2)}}&\lvert \P^{(2)}\{b+1:L(\mu^{(2)})+1\}\uparrow^{\hat{\mu}^{(2)}} \rvert  \\
    &+\sum_{\mu^{(2)} > \mu > {\mu'}^{(2)}} 2^{\mu} W^{(2)}(\P^{(2)},\mu)  \\
    &+2^{{\mu'}^{(2)}}\lvert \P\{L({\mu'}^{(2)})+1:b'+1\}\uparrow^{{\mu'}^{(2)}} \rvert \leq M.
\end{align*}
\end{lemma}
The above two equations essentially convey that the weight of the adversary's trimmed chain $\P^{(2)}$ between $b$ and $b'$ is upper bounded by $M$.
\begin{proof}[Proof of Lemma \ref{lm:bounding_adversary_work}]
This is a direct consequence of the dominant-superchain property.
\end{proof}

Next, we upper bound the weight of a trimmed chain.

\begin{lemma}
Let $\C$ be a blockchain and $\P$ be its corresponding trimmed chain produced by Algorithm \ref{algo:trimming}. Let $\delta \leq 1/\sqrt{2}$. If $\lvert \C \rvert = B$, then except with probability $\frac{\negl(k)}{B^{a \delta^2/2 - 1}}$, $S(\P,0) + \lvert \P\{B':\} \rvert \leq ((1+\delta)^2 + \delta)B$.
\label{lm:trimmed_chain_upperbound}
\end{lemma}
\begin{proof}
Let $\mu$ be the highest level with $S(\P,\mu)\geq \delta B$. This means that, for all $\mu' \leq \mu$,
\begin{equation*}
    \lvert \P\{L_f(\mu'):L_l(\mu')+1\}\uparrow^{\mu'} \rvert \geq k + a \log(\delta B).
\end{equation*}
By Lemma \ref{lm:poisson_length_compare}, we get that for every $\mu' \leq \mu$,
\begin{equation*}
\lvert \P\{L_f(\mu'):L_l(\mu')+1\}\uparrow^{\mu'} \rvert
     \leq (1+\delta^2) \left( L_l(\mu')-L_f(\mu') + 1 \right),
\end{equation*}
except with probability 
\begin{equation*}
    2 e^{- \frac{k + a \log(\delta B)}{2} {\delta}^2} = \frac{\negl(k)}{B^{a\delta^2/2}}.
\end{equation*}

From Lemma \ref{lm:max_level}, the number of level-ranges in $\P$ is at most $\log B/\log 2$, except with probability $\negl(k)/B^{a/4}$. Therefore, by a union bound,
\begin{equation*}
    S(\P,0) + \lvert \P\{B':\} \rvert \leq ((1+\delta)^2) + \delta)B,
\end{equation*}
except with probability,
\begin{equation*}
    \log B \cdot \frac{\negl(k)}{B^{a\delta^2/2}} + \frac{\negl(k)}{B^{a/4}} = \frac{\negl(k)}{B^{a \delta^2/2 - 1}}.
\end{equation*}
\end{proof}

Next, we prove Theorem \ref{th:trim_attacked}.

\subsection{Proof of Theorem 1}

\begin{proof}[Proof of Theorem \ref{th:trim_attacked}]
The proof is by induction on the length of the honest blockchain $B^{(1)}$.

As the base case, consider $B^{(1)} = ck$. Since no trimming has occurred up to this point, the probability of a trim-attack up to this point is 0.

As the induction hypothesis, assume that up to chain-length $B$, the adversary wasn't able to create a fork from beyond a honest node's trimming point. That is,
\begin{equation}
\begin{split}
    \forall t &\text{ such that } B^{(1)}_t < B, \mathsf{Compare}\left( \P^{(1)}_t, \P^{(2)}_t \right) = \P^{(2)}_t 
    \text{ implies that } \\ & \mathsf{LCA}\left( \P^{(1)}_t, \P^{(2)}_t \right) \geq B^{(1)}_t - \Delta(S(\P^{(1)}_t,0) + \lvert \P\{{B'}^{(1)}_t:\} \rvert). 
\end{split}
\end{equation}

For the induction step, consider a time $t$ when the length of the honest chain is $B_t^{(1)}=B$. As a contradiction, assume that the adversary has created a secret chain $\P^{(2)}_t$ such that,
$\mathsf{Compare}\left(\P^{(1)}, \P^{(2)}\right) = \P^{(2)}$, and $b = \mathsf{LCA}\left( \P^{(1)}, \P^{(2)} \right) < B - \Delta(S(\P^{(1)}_t,0) + \lvert \P\{{B'}^{(1)}_t:\} \rvert)$.

We now split into two cases based on the location of $b$:

\paragraph{Case 1: $b \geq \delta B$}

By superchain quality of $\P^{(1)}$, we have, 
\begin{equation}
    \begin{split}
        S(\P^{(1)}, \mu^{(1)}) &\geq  (1 - \delta)b, \\
        &\geq (1-\delta)\delta B.
    \end{split}
\end{equation}

Moreover, since $\P^{(2)}$ shares blocks with $\P^{(1)}$ up to block $b$, we have, 
\begin{equation}
\begin{split}
    S(\P^{(2)}, \mu^{(2)}) &\geq  (1 - \delta)b, \\
    &\geq (1-\delta)\delta B.
\end{split}
\end{equation}

As a consequence of Lemma \ref{lm:bounding_adversary_work}, the weight of $\P^{(2)}$ between blocks $b$ and $b'$ is upper bounded by $M$. 

Next, we define a block $b'' \geq b'$ which is such that all level ranges in $\P^{(2)}\{b'':\}$ are of size at least $c'g(\P^{(1)}, \mu^{(1)})$, for some $1 < c' < c$. In order to ensure this, $b''$ can be defined as,
\begin{equation}
    b'' = \begin{cases}
    b', & \text{if } \lvert \P^{(2)}\{b'+1:L_l({\mu'}^{(2)})+1 \}\uparrow^{{\mu'}^{(2)}} \rvert \\
    &\quad \geq c'(k + a\log((1-\delta)\delta B)), \\
    L^{(2)}_l({\mu'}^{(2)}) +1, &\text{otherwise},
    \end{cases}
\end{equation}
where ${\mu'}^{(2)}$ is the level range in which block $b'$ resides in $\P^{(2)}$. If $b' > b$, then ${\mu'}^{(2)} \leq \max \{{\hat{\mu}}^{(1)}, \mu^{(1)}-1\}$. Therefore, the weight of $\P^{(2)}$ between $b'$ and $b''$ is upper bounded by $c'M$.

Let $b^*$ be the earliest block such that all blocks in $\P^{(2)}\{b^*+1:\}$ were mined by the adversary. Then, by the induction hypothesis,
\begin{equation*}
    b^* - b'' \leq \Delta(S(\P^{(1)}_t,0)+\lvert \P^{(1)}_t\{{B'}^{(1)}:\} \rvert).
\end{equation*}
The above statement is proved as follows. Assume there was a time $t^*$ when $b^*$ was in a honest node's chain, but at a later time $t$ the block after $b''$ in $\P^{(2)}_t$ is not in the honest node's chain $\C_t$. By the induction hypothesis,
\begin{equation*}
\begin{split}
    b^* - b'' &\leq \Delta(S(\P^{(1)}_{t^*},0)+\lvert \P^{(1)}_{t^*}\{{B'}^{(1)}:\} \rvert), \\
    &\leq \Delta(S(\P^{(1)}_t,0)+\lvert \P^{(1)}_t\{{B'}^{(1)}:\} \rvert).
    \end{split}
\end{equation*}

Given, $k' = k - a \log\left(((1+\delta^2) + \delta)/\delta \right)$. Therefore, from Lemma \ref{lm:trimmed_chain_upperbound}, except with probability $\frac{\negl(k)}{B^{a \delta^2/2 - 1}}$,
\begin{equation*}
    \Delta(S(\P^{(1)}_t,0)+\lvert \P^{(1)}_t\{{B'}^{(1)}:\} \rvert) \leq g(\P^{(1)}, \mu^{(1)}).
\end{equation*}

Therefore,
\begin{equation*}
    \mathsf{Weight}(\P^{(2)},b^*) \geq \frac{c'-1}{c'}\frac{c-1-c'}{c} \mathsf{Weight}(\P^{(1)},b).
\end{equation*}

Since, each level range in $\P^{(2)}\{b^*:\}$ contains at least $g(\P^{(1)},\mu^{(1)})$ number of corresponding superblocks, we have that except with probability $\negl(k) \log B/B^{a\delta^2/2}$, the time required $T$ required by the adversary is lower bounded as,
\begin{equation*}
    T \geq \frac{(1-\delta)}{\lambda_a}\frac{c'-1}{c'}\frac{c-1-c'}{c} \mathsf{Weight}(\P^{(1)},b).
\end{equation*}

In this time, except with probability $\negl(k)/B^{a\delta^2/2}$, the honest chain growth is lower bounded as,
\begin{equation*}
    \lvert \C\{b:\} \rvert \geq (1-\delta)^2\frac{\lambda_h}{\lambda_a}\frac{c'-1}{c'}\frac{c-1-c'}{c} \mathsf{Weight}(\P^{(1)},b).
\end{equation*}

We know from superquality property that $\mathsf{Weight}(\P^{(1)},b) \geq (1-\delta)\lvert \C\{b:\} \rvert$. Therefore, we reach a contradiction if, 

\begin{equation}
 (1-\delta)^3\frac{\lambda_h}{\lambda_a}\frac{c'-1}{c'}\frac{c-1-c'}{c}
 > 1.
 \label{eq:case_1_conclusion}
\end{equation}


\paragraph{Case 2: $b < \delta B$}
The reason to consider this case separately is because when $b < \delta B$, the corresponding level-range containing $b$ may not have enough number of blocks for us to obtain a strong concentration bound. Therefore, we jump ahead to a level which contains enough superblocks. Let ${\tilde{\mu}}^{(2)}$ be the highest-level in $\P^{(2)}$ such that, 
\begin{equation*}
    S(\P^{(2)},{\tilde{\mu}}^{(2)}) \geq \delta B.
\end{equation*}

In case $b'$ falls in the level range of $\tilde{\mu}^{(2)}$ or below, then the same analysis for Case 1 holds. In particular, if the relation in \eqref{eq:case_1_conclusion} is satisfied, we reach a contradiction except with probability $\frac{\negl(k)}{B^{a \delta^2/2 - 1}}$.

Otherwise, if $b'$ falls in a level-range higher than $\tilde{\mu}^{(2)}$, we set,
\begin{equation*}
    b'' = L^{(2)}_f(\tilde{\mu}^{(2)}).
\end{equation*}

Since $b<\delta B$, we have,

\begin{equation*}
    \begin{split}
        \mathsf{Weight}(\P^{(2)}, b'') &\geq (1-\delta)\mathsf{Weight}(\P^{(2)}, b), \\
        &\stackrel{(x)}{\geq} (1-\delta)\mathsf{Weight}(\P^{(1)}, b), \\
        &\stackrel{(y)}{\geq} (1-\delta)^3B.
    \end{split}
\end{equation*}
(x) follows because we assume that $\P^{(2)}$ beats $\P^{(1)}$. (y) follows by superchain quality of the honest chain and the fact that $b< \delta B$.

By the same argument of the induction hypothesis used in Case 1, except with probability $\frac{\negl(k)}{B^{a \delta^2/2 - 1}}$, $(c-1)/c$ fraction of the blocks in $\P^{(2)}\{b'':\}$ were mined by the adversary. Again, except with probability $\frac{\negl(k)}{B^{a \delta^2/2}}$, by the time that the adversary mines those blocks, the length of the honest chain would be,
\begin{equation*}
    \lvert \C\{b+1:\} \rvert \geq (1-\delta)^5\frac{\lambda_h}{\lambda_a} \frac{c-1}{c}B. 
\end{equation*}

The above equation is a contradiction if,
\begin{equation}
 (1-\delta)^5\frac{\lambda_h}{\lambda_a} \frac{c-1}{c} > 1.
 \label{eq:case_2_conclusion}
\end{equation}
\noindent \rule{1cm}{0.4pt} End Case 2 \rule{1cm}{0.4pt}

There are at most $B - \Delta(\P^{(1)})$ unique choices that can be made for block $b$. For each choice, the conclusions of Case 1 and Case 2 hold except with probability $\frac{\negl(k)}{B^{a \delta^2/2-1}}$. Therefore, taking a union bound over all choices of $b$, the conclusions hold except with probability $\frac{\negl(k)}{B^{a \delta^2/2-2}}$. This completes the induction step.

The induction step at chain-length $B$ is violated with probability $\frac{\negl(k)}{B^{a \delta^2/2-2}}$. Therefore, summing over all chain lengths, we have,
\begin{equation}
    \begin{split}
        \Pr(\mathsf{trim-attacked}) &\leq \sum_{B=1}^\infty \frac{\negl(k)}{B^{a \delta^2/2-2}}, \\
        &= \negl(k),
    \end{split}
\end{equation}
where the last equation follows since $a \delta^2/2-2 \geq 2$.

\end{proof}

\subsection{Proofs of Other Theorems:}

\begin{proof}[Proof of Corollary \ref{corr:congruence}]
    We have proved in Theorem \ref{th:trim_attacked} that the adversary cannot create a longer fork from before a honest node's trimming point except with probability $\negl(k)$. Since a honest node retains all blocks after the trimming point, comparing two trimmed chains that differ only after the trimming point is equivalent to comparing the underlying chains (see line 11 and 12 in Algorithm \ref{algo:compare}).
    \end{proof}
    
\begin{proof}[Proof Sketch of Theorem \ref{th:state_attacked}]
    In order to pass through the $\mathsf{state-verify}$ protocol in Algorithm \ref{algo:state-verify}, the adversary has to produce a malicious state at the trimming point of an honest node, such that it verifies against all the states corresponding to the untrimmed chain tail. In Section \ref{ssec:stochastic_model}, we explained that the adversary's process of generating a malicious state sequence is a Poisson process of rate $\lambda_s$, where $\lambda_s << \lambda_a$. Therefore, from the same argument as in Theorem \ref{th:trim_attacked}, we can prove that there exists no time when the adversary can create such a malicious state.
    \end{proof}
    
    \begin{proof}[Proof of Corollary \ref{corr:bootstrap_attacked}]
    This is a corollary of Theorems \ref{th:trim_attacked} and \ref{th:state_attacked}. Since, the probability that there exists a time when the system is either \\ $\mathsf{trim-attacked}$ or $\mathsf{state-attacked}$ is very small, the probability that an arriving node sees the system attacked in either way is negligible in $k$.
    \end{proof}
    
    \begin{proof}[Proof Sketch of Theorem \ref{th:succinctness}]
    Assume that all relevant quantities concentrate strongly around their means. First, the highest level-range would be $\mu_h = O(\log B)$. This is because, a level $\mu_h$ is $2^{\mu_h}$ times harder to find than a level-0 superblock.
    
    Next, a level range of level-$\mu$ would contain at most $O(\log B)$ $\mu$-superblocks. Because, if there were more, then it could be trimmed to level-$(\mu+1)$. Since we also retain suffixes of length $O(\log B)$ blocks of every lower level, level-range $\mu$ would contain about $O(\log^2 B)$ blocks. Adding across all the $\mu_h$ number of level ranges, the trimmed portion of $\P$ would contain about $O(\log^3 B)$ blocks. And, the untrimmed tail would contain at most $O(\log B)$ blocks.
    
    Since an interlink contains a link to the previous superblock of every level, the size of the interlink would be $O(\log B)$, thus making the size of each block $O(\log B)$. Therefore, the total storage requirement would be $O(\log^4 B)$. These statements are proved more rigorously in appendix \ref{app:succinctness}.
    \end{proof}

\section{Trimming Attack}
\label{app:trimming-attack}

Here we prove Theorem \ref{th:trimming_attack}.

\begin{proof}[Proof of Theorem \ref{th:trimming_attack}]
Let $f(x)=c_0 \log(1+x)$, for some small constant $c_0$. Define a sequence $\left( \alpha_i \right)_{i}$ as follows:

\begin{equation*}
    \begin{split}
        \alpha_0 &= 0, \\
        \alpha_1 &= 1, \\
        \alpha_{i} &= 2f\left( \sum_{j=1}^{i-1} \alpha_j \right), \quad i\geq 2.
    \end{split}
\end{equation*}

The adversarial attack follows a sequence of attempts. The $i\textsuperscript{th}$ attempt begins when $\lvert \C_{t} \rvert = \sum_{j=0}^{i-1} \alpha_j$. At this point, the adversary abandons the previous attempt and starts mining on a secret chain $\D_t$ forking away from $\C_t$ at this tip. The attempt lasts until one of $\C_t$ or $\D_t$ reach a length of $\left(\sum_{j=0}^{i} \alpha_j\right) + 1$. If $\D$ reaches the length first, then the adversary waits until $\C$ reaches a length of $\sum_{j=0}^{i} \alpha_j$, at which point it publishes $\D$. Since, $\D$ will be longer than $\C$, the honest nodes will adapt $\D$. The fact that this is a successful trimming attack follows from the following claim:

\begin{claim1}
For all large enough $i$,
$\alpha_i \geq f(\sum_{j=1}^{i} \alpha_j)$
\end{claim1}

From the definition of the attack, the length of the honest chain when adversary publishes his secret chain is $\sum_{j=1}^i \alpha_i$. Recall, $f(\sum_{j=1}^i \alpha_i)$ is an upperbound on the number of blocks between the trimming point and the tip of the chain for the honest node, since $B-B' = o(\log B)$. Therefore, the above claim shows that the adversary has successfully created a fork from before a trimming point of any node. To prove the claim, for large enough $i$ we have,
\begin{align*}
    \alpha_i &= f\left( \sum_{j=1}^{i-1} \alpha_j \right) + f\left( \sum_{j=1}^{i-1} \alpha_j \right), \\
        &\stackrel{(a)}{\geq} f\left( \sum_{j=1}^{i-1} \alpha_j \right) + f(\alpha_i), \\
        &\stackrel{(b)}{\geq} f\left( \sum_{j=1}^{i} \alpha_j \right).
\end{align*}
(a) follows because $2c_0\log(z) \leq z$. (b) follows by the sub-additivity of $\log(1+z)$.

The below two claims are needed to compute the probability of a successful attack attempt.

\begin{claim1}
$\alpha_i \leq c_1 \log i$, for all $i \geq 2$ and $c_1 = 1+4c_0$.
\label{claim:log_upper_bound}
\end{claim1}
Assuming $c_0 \leq 1/4$, $1 \leq c_1 \leq 2$. We will prove the above claim by induction. As the base case, $\alpha_2 = 2c_0\log(1+\alpha_11) = 2c_0\log2 \leq c_1\log2$. Assume that the claim is true for all $j\leq i-1$. Then,
\begin{align*}
    \alpha_i &= 2c_0\log\left( 1 + \sum_{j=1}^{i-1} \alpha_j \right), \\
    &\leq 2c_0 \log(1+(i-1)c_1\log i), \\
    &\stackrel{(a)}{\leq} 2c_0\log(i c_1 \log i), \\
    &= 2c_0 \log i + 2c_0 \log c_1 + 2c_0 \log\log i, \\
    &\stackrel{(b)}{\leq} 2c_0 \log i + 2c_0 \log i + \log i, \\
    &= (4c_0 + 1) \log i.
\end{align*}
(a) follows since $c_1 \geq 1$. (b) follows since $i \geq 2 \geq c_1$.

The probability of the adversary's $i\textsuperscript{th}$ attempt succeeding is $e^{-c_2 \alpha_i}$, for some constant $c_2$ \cite{sompolinsky2016bitcoin}.
\begin{claim1}
$\sum_{i=1}^{\infty} e^{-c_2\alpha_i} = \infty$.
\label{claim:infinite_sum}
\end{claim1}
To prove the claim,
\begin{align*}
    \sum_{i=2}^{\infty} e^{-c_2\alpha_i} &= \sum_{i=2}^\infty e^{-2c_0c_2 \log (1+\sum_{j=1}^{i-1} \alpha_i)}, \\
    &\stackrel{(a)}{\geq} \sum_{i=2}^\infty e^{-\log (1+\sum_{j=1}^{i-1} \alpha_j)}, \\
    &= \sum_{i=2}^\infty \frac{1}{1+\sum_{j=1}^{i-1} \alpha_j}, \\
    &\stackrel{(b)}{\geq} \sum_{i=2}^\infty \frac{1}{1+\sum_{j=1}^{i-1} c_1 \log i}, \\
    &\stackrel{(c)}{=} \sum_{i=2}^\infty \frac{1}{1+ c_1 i \log i}, \\
    &= \infty.
\end{align*}
(a) follows since $c_0$ can be made as small as desired. (b) follows from Claim \ref{claim:log_upper_bound}. (c) can be verified to diverge using the integral test. 

Since the attack-attempts are carried over disjoint sections of the blockchain, all the attempts are mutually independent. Therefore, from Claim \ref{claim:infinite_sum} and the second Borel-Cantelli Lemma, we have,
\begin{equation*}
    \Pr(\mathsf{trim-attacked}) = 1.
\end{equation*}
\end{proof}

\section{Dominant Superchain}
\label{app:dominant_superchain}
We will use the following concentration bound for sums of independent, but not necessarily identically distributed, Bernoulli random variables. It can be proved by an elementary application of Chernoff bound and can be found in  \cite{dubhashi1998}.

\begin{lemma}
Let $\{ X_i \}_{i=1}^n$ be a sequence of independent random variables with $X_i \sim Bernoulli(p_i)$. Let $S_n = \sum_{i=1}^n X_i$, and $s_n = \sum_{i=1}^n p_i$. Then, for $0 < \varepsilon < 1$,
\begin{equation*}
    \Pr\left( S_n \leq (1-\varepsilon)s_n \right) \leq e^{-\frac{\varepsilon^2}{2}s_n}. 
\end{equation*}
\label{lm:chernoffbernoulli}
\end{lemma}

We now prove the high likelihood of dominant superchains.

\begin{lemma}
Let $\C$ be a blockchain and $\C\uparrow^{\mu}$ be its corresponding level-$\mu$ upchain. Then, for any $k\in \mathbb{N}$ and $g \geq k + a\log(\lvert \C \rvert)$. Then, if $a\geq 8$, $\C\uparrow^{\mu}$ is a dominant superchain with high probability in $k$ and $\lvert C \rvert$.
\end{lemma}
\begin{proof}
Consider a starting block $b_1 \geq \C[0]$ and ending block $b_2 \leq \C[-1]$.

First, observe that if $\C[b_1:b_2+1]\uparrow^{\mu} \geq 1$, then any trimmed chain $\P$ starting at $b_1$ and ending at $b_2$, and with its highest level range being at most $\mu$, will contain at least one level-$\mu$ superblock. This is because a level-$\mu$ superblock, call it $b_3$, is also a level-$\mu'$ superblock for any $\mu' \leq \mu$. Therefore, the interlink data-structure ensures that a there is a level range $\mu'$, which contains a link to $b_3$.

Let $\P'$ be the trimmed chain of the highest weight starting at $b_1$ and ending at $b_2$, and with its highest level-range being at most $\mu$. Assume that it satisfies the premise as defined in equation \ref{eq:dom_supchain_2}. That is,
\begin{equation*}
\begin{split}
    2^{\mu_1}\lvert \P'\{l:&L'_l(\mu')+1\}\uparrow^{\mu_1} \rvert \\
    + &\sum_{\mu'' < \mu'} 2^{\mu''}\lvert \P'\{L'_f(\mu''):L'_l(\mu'')+1]\uparrow^{\mu''} \rvert \geq 2^{\mu}g.
    \end{split}
\end{equation*}

It is useful re-arrange the equation by dividing both sides by $2^\mu$.
\begin{equation}
\begin{split}
    2^{\mu_1-\mu}\lvert \P'\{l:&L'_l(\mu')+1\}\uparrow^{\mu_1} \rvert \\
    + &\sum_{\mu'' < \mu'} 2^{\mu''-\mu}\lvert \P'\{L'_f(\mu''):L'_l(\mu'')+1]\uparrow^{\mu''} \rvert \geq g.
    \end{split}
    \label{eq:dom_superchain_app_1}
\end{equation}
Given that a block is a level-$\mu''$ superblock, let $X$ be an indicator variable for it also to be a level-$\mu$ superblock. Then, $\Pr(X=1) = 2^{\mu''-\mu}$. Therefore, we can associate a sequence of Bernoulli random variables $\{X_i\}_{i=1}^n$ with $\P$, where $X_i$ indicates whether the the $i\textsuperscript{th}$ superblock as taken in \eqref{eq:dom_superchain_app_1} is a level-$\mu$ superblock. Denoting $\Pr(X_i=1)= p_i$, equation \eqref{eq:dom_superchain_app_1} says that,
\begin{equation*}
    s_n = \sum_{i=1}^n p_i \geq g.
\end{equation*}

Now, using Lemma \ref{lm:chernoffbernoulli}, we can lower bound $S_n = \sum_{i=1}^n X_i$, as,
\begin{equation*}
\begin{split}
    \Pr(S_n \leq 1) &\leq e^{-\frac{(1-1/g)^2}{2}g}, \\
    &\leq e^{-g/4}, \\
    &\leq \frac{\negl(k)}{B^{a/4}},
\end{split}
\end{equation*}
where we denote $B = \lvert C \rvert$. There are totally $B$ ways to choose $b_1$, and then $B-b_1$ ways to choose $b_2$. Therefore, there are totally at most $B^2$ ways to choose $b_1$ and $b_2$. Taking a union bound over all possible starting points $b_1$ and ending points $b_2$, we conclude there exists at least one level-$\mu$ superblock in $\C$ with high probability if $a \geq 8$.
\end{proof}

\section{Succinctness}
\label{app:succinctness}
\begin{lemma}
When the length of the blockchain $\C$ is $B$, the highest-level of superblock in the corresponding trimmed chain $\P$ is at most $\log B / \log 2$ except with probability $\negl(k)/B^{a/4}$. 
\label{lm:max_level}
\end{lemma}
\begin{proof}
Let $\mu = \log B / \log 2$. Then, weight of 1 level-$\mu$ superblock is $2^{\mu} = B$. Therefore, $g(\P,\mu) \geq B$.

Also, $\lvert \C\uparrow^\mu \rvert \sim Poisson(1)$. Therefore, applying the Poisson concentration bound,

\begin{align*}
    \Pr\left(\lvert \C\uparrow^\mu \rvert \geq g(\P,\mu) \right) &\leq \Pr\left(\lvert \C\uparrow^\mu \rvert \geq 1 + (k-1) + a\log B \right), \\
    &\leq e^{-\frac{k-1 + a\log B}{4}}, \\
    &= \frac{\negl(k)}{B^{a/4}}.
\end{align*}
The proof is complete by observing that if $\lvert \C\uparrow^\mu \rvert < g(\P,\mu)$, then the chain won't be trimmed to this level.
\end{proof}

Next, we show that individual level ranges aren't too big.

\begin{lemma}
Let $\P$ be a trimmed chain of an underlying chain $\C$ obtained by the trimming protocol. Let $L_f$ and $L_l$ be the associated level range functions. Then, for every level $\mu$, $\lvert \P\{L_f(\mu)+1:L_l(\mu)\}\uparrow^\mu \rvert \leq 4 f(\P,\mu)$, with high probability.
\label{lm:small_level_ranges}
\end{lemma}
\begin{proof}
As a contradiction, assume that there were a level $\mu$, such that $\lvert \P\{L_f(\mu)+1:L_l(\mu)\}\uparrow^\mu \rvert \leq 4 f(\P,\mu)$. Then, from Lemma \ref{lm:poisson_length_compare}, with high probability, we would have $\lvert \P\{L_f(\mu)+1:L_l(\mu)\}\uparrow^{\mu+1} \rvert \leq  f(\P,\mu)$. This means that the trimming protocol would have trimmed this range to level $\mu+1$. 
\end{proof}

Now we are ready to prove the optimistic succinctness theorem.

\begin{proof}[Proof Theorem \ref{th:succinctness}]
From Lemma \ref{lm:max_level} we know that the highest level is $O(\log B)$ with high probability. And, from Lemma \ref{lm:small_level_ranges}, the number of $\mu$-level superblocks in level range $\mu$ is $O(\log B)$ with high probability.

According to the construction of level range $\mu$ in the trimming protocol, it contains $O(\log B)$ superblock of level-$\mu'$ for each level $\mu' < \mu$. Therefore, each level range contains $O(\log^2 B)$ blocks with high probability. And, since there are at most $O(\log B)$ level ranges, the number of blocks in the trimmed section of $\P$ is $O(\log^3 B)$ blocks. Since, the untrimmed tail consists of $O(\log B)$ blocks, $\P$ totally contains $O(\log^3 B)$ blocks.

Each block has the interlink, which in turn contains a link to the previous superblock of every level. Since there are at most $O(\log B)$ level of superblocks, the size of the interlink is $O(\log B)$. This makes the overall size of $\P$ as $O(\log^4 B)$.
\end{proof}

\section{Our Protocol in the Random Oracle Model} 
\label{ssec:pass}

In our analysis, we have modelled PoW mining as calls to an ideal random functionality which returns a ``success'' with a small probability, and ``failure'' otherwise, independent of everything else. However, in current PoW systems, mining is performed by making calls to a random oracle hash function, $H: \{0,1\}^* \to \{0,1\}^\kappa$, where $\kappa$ is a security parameter. Since the output space of $H$ is finite, we cannot have our formulation of $\mathsf{trim-attacked}$ in Definition \ref{def:trim_attacked} as a security model. This is because $\mathsf{trim-attacked}$ considers time from $0$ to infinity, during which time the hash function $H$ is going to experience collisions with probability 1. In particular, in order to ensure that $H$ doesn't experience collisions with high probability, one can only run the blockchain for $\mathsf{poly}(\kappa)$ time (or, $\mathsf{poly}(\kappa)$ rounds in the discrete time model)\footnote{$\mathsf{poly}(\kappa)$ means polynomial in $\kappa$.}.

In \cite{pass2017}, the authors denote our ideal random functionality as $\mathcal{F}_{tree}$. They prove that when a blockchain is run for $\mathsf{poly}(\kappa)$ rounds, the $\mathcal{F}_{tree}$ model and the random oracle model are statistically close (see Lemma 5.1 in \cite{pass2017}). Using this, they prove that if a blockchain protocol is secure under the $\mathcal{F}_{tree}$ model, it is secure under the random oracle model for $\mathsf{poly}(\kappa)$ number of rounds as well. We can employ this result in transferring our security definitions in Section \ref{sec:security-statements} to the more realistic random oracle model for hashing. For example, Theorem \ref{th:trim_attacked} would imply that an adversary without a majority of the mining power cannot create a longer-fork from before a honest node's trimming point at any point during the $\mathsf{poly}(\kappa)$ rounds of blockchain execution (instead of infinite number of rounds). Similar arguments hold under the Bitcoin Backbone model\cite{kiayias2016} as well.

\section{Optimizing Storage with Stateless Blockchains}
\label{sec:app-stateless_blockchains}
Our trimming protocol optimizes the amount of cold storage required to represent the PoW in the blockchain.
In this section, we describe how our work can be interfaced to other methods that optimize the storage of the \emph{state}.
In hybrid nodes, states need to be stored both in cold and hot storage.
The state at the trimming point, $\S(B'_t)$, needs to be stored in cold-storage in order to recover the state at the chain tip when forks are created.
Conversely, the state at the tip of the chain, $\S(B_t)$, needs to be stored in hot-storage in order to perform transaction and block validation.
Along with the length of the blockchain, the size of the blockchain-state is also rapidly increasing with time.
For instance, at the time of writing Bitcoin's UTXO state is almost 4GB in size, which motivates lessening high storage and verification time burdens by extending our work to optimize both the state's cold-storage at the trimming point, and its hot-storage at the chain tip.

So far, we have described that storing state-commitments in the blocks enables one to securely establish the state of a trimmed chain, by use of the $\mathsf{state\mhyphen verify}$ protocol described in Algorithm \ref{algo:state-verify}.
In a different line of work, called \emph{stateless blockchains}, state-commitments  are used to reduce the amount of hot-storage required for transaction validation, thus speeding up validation. In this section, we illustrate that when hybrid nodes are used in stateless-blockchains, then state commitments can serve a dual purpose: 1) they can be used to establish the state of the trimmed chain; 2) and, they can be used to perform stateless transaction validation.

Stateless validation, first proposed by Todd \cite{todd}, is a scheme where nodes validating transactions store a short cryptographic state-commitment rather than the entire state. A client then provides a membership proof that the node can verify against the commitment. Recently, several constructions have been developed to perform stateless validation in both the UTXO and account-based state models using various cryptographic primitives \cite{chepurnoy2018,dryja2019,boneh2019, tomescu2020, reyzin2017, agrawal2020, gorbunov2020}.
As an example, the following is how stateless validation of an account-based state would work via \emph{vector commitments}, a primitive allowing one to commit to an ordered sequence of values in a way that allows one to later open the commitment at a specific position in the sequence, with proof that the $i^{th}$ value is the $i^{th}$ committed message (see \cite{catalano2013vector}). Suppose an account is represented by a tuple of the form mapping a public key to a balance:
$(pk_i,v_i)$, and that party 1 want to transfer an amount $m$ to party 2. Then, the transaction submitted by party 1 takes the form $tx = (pk_1,v_1,pk_2,m,\pi,\sigma)$ where $pk_1$ is the public key of the source account, $pk_2$ is the public key of the destination account, $v_1$ is the initial balance of $pk_1$, $m$ is the amount being transferred, $\pi$ is a proof of membership of $(pk_1,v_1)$ in the commitment, and $\sigma$ is a signature of $(v_1,pk_2,m,\pi)$ under $pk_1$. 
In order to validate a transaction, a node checks that $m <v_1$, the validity of the signature, and the validity of the proof with respect to the public commitment $y(B_t)$. If the checks passes, the commitment is updated to reflect the change in the two accounts: $(pk_1, -m)$ and $(pk_2, +m)$. With each such update, users update their proofs $\pi$ so that they remain valid with respect to the latest commitment. 

As a first step, observe that stateless blockchains optimize hot-storage by obviating the need to store the state at the tip of the chain in the RAM. At this point, the state at the trimming point is still stored in the cold storage. We can use this to provide an interface from trimming to stateless blockchains: Firstly, if the proof for some client becomes outdated past the trimming point because it was offline for a long while, then the client can recompute its latest proof by querying a hybrid node for $\S(B'_t)$ and the untrimmed chain tail. Secondly, if a fork in the chain makes the proofs of several clients invalid, they can all recompute their proofs by contacting the hybrid node similarly.

Now, we describe a way to optimize the cold-storage of the state at the trimming point as well. In order to do this, we need to make an assumption on the clients. At every point, the client needs to store its membership proof corresponding to every block in the untrimmed tail of hybrid nodes. This enforces a couple of constraints on clients. First, the storage requirements of a client would scale logarithmically in the length of the blockchain as well, since the length of the chain-tail grows logarithmically. Second, the clients need to be online very often so that their proofs never get outdated.
If this is too cumbersome, clients can alternatively delegate the job of saving and updating proofs to proof-serving nodes \cite{chepurnoy2018}. In this case, the hybrid node need not store the entire state at the trimming point either. This is because, even in the event that a fork is created in the untrimmed-tail, all the clients would have their proofs corresponding to the forking point. In this case, a joining hybrid node would no longer need to use the $\mathsf{state\mhyphen verify}$ protocol, since just the state commitment suffices. The rest of the hybrid node's protocols proceed as before.

\end{document}